\documentclass{article}
\usepackage{amsmath,amssymb,amsfonts}

\def\op#1{\mathop{{\it\fam0} #1}\limits}

\newcommand{\beq}{\begin{equation}}
\newcommand{\eeq}{\end{equation}}
\newcommand{\ben}{\begin{eqnarray}}
\newcommand{\een}{\end{eqnarray}}
\newcommand{\be}{\begin{eqnarray*}}
\newcommand{\ee}{\end{eqnarray*}}
\newcommand{\bea}{\begin{eqalph}}
\newcommand{\eea}{\end{eqalph}}

\newcommand{\gG}{{\mathfrak g}}

\newcommand{\di}{{\mathrm {dim}\,}}

\newcommand{\pr}{{\mathrm{pr}\,}}

\newcommand{\al}{\alpha}
\newcommand{\bt}{\beta}
\newcommand{\dl}{\delta}
\newcommand{\la}{\lambda}
\newcommand{\La}{\Lambda}
\newcommand{\f}{\phi}

\newcommand{\F}{\Phi}
\newcommand{\p}{\pi}

\newcommand{\Om}{\Omega}
\newcommand{\m}{\mu}

\newcommand{\g}{\gamma}
\newcommand{\G}{\Gamma}

\newcommand{\thh}{\theta}
\newcommand{\vt}{\vartheta}
\newcommand{\cG}{{\mathfrak g}}
\newcommand{\ve}{\varepsilon}
\newcommand{\up}{\upsilon}

\newcommand{\ap}{\approx}

\newcommand{\si}{\sigma}
\newcommand{\Si}{\Sigma}

\newcommand{\cJ}{{\mathcal J}}

\newcommand{\gE}{{\mathfrak E}}

\newcommand{\cL}{{\mathcal L}}
\newcommand{\cV}{{\mathcal V}}

\newcommand{\cE}{{\mathcal E}}
\newcommand{\cH}{{\mathcal H}}

\newcommand{\bL}{{\mathbf L}}

\newcommand{\w}{\wedge}

\newcommand{\wt}{\widetilde}
\newcommand{\wh}{\widehat}
\newcommand{\ol}{\overline}

\newcommand{\dr}{\partial}
\newcommand{\rrq}{{\ol q}}

\newcommand{\Ker}{\mathrm{Ker}\,}

\newcommand{\ar}{\op\longrightarrow}

\newcommand{\ot}{\otimes}

\let\ssection=\section
\renewcommand{\section}{\setcounter{equation}{0}\ssection}

\newenvironment{eqalph}{\stepcounter{equation}
\setcounter{equationa}{\value{equation}} \setcounter{equation}{0}
\begin{eqnarray}}{\end{eqnarray}\setcounter{equation}{\value{equationa}}}

\newcounter{equationa}[section]
\newcounter{remark}[section]
\newcounter{example}[section]
\newcounter{theorem}[section]
\newcounter{condition}[section]
\newcounter{lemma}[section]
\newcounter{corollary}[section]
\newcounter{definition}[section]
\def\theremark{\arabic{section}.\arabic{remark}}

\def\thedefinition{\arabic{section}.\arabic{theorem}}

\newenvironment{proof}{{\it Proof.}}{\hfill $\Box$
\medskip }
\newenvironment{remark}{\refstepcounter{remark} \medskip {\bf Remark
\theremark.} }{ \medskip }
\newenvironment{example}{\refstepcounter{remark} \medskip {\bf
Example \theremark.} }{ \medskip }
\newenvironment{theorem}{\refstepcounter{theorem} \medskip{\bf
Theorem \thedefinition.}\it }{ \medskip }

\newenvironment{definition}{\refstepcounter{theorem} \medskip{\bf
Definition \thedefinition.} \it}{ \medskip }

\newcommand{\mar}[1]{}

\begin{document}

\hbox{}

\begin{center}

{\Large\bf Noether's first theorem in Hamiltonian mechanics}

\bigskip

G. Sardanashvily

\medskip

Department of Theoretical Physics, Moscow State University, Russia

Lepage Research Institute, Czech Republic

\bigskip

\end{center}

\begin{abstract}
Non-autonomous non-relativistic mechanics is formulated as
Lagrangian and Hamiltonian theory on fibre bundles over the time
axis $\mathbb R$. Hamiltonian mechanics herewith can be
reformulated as particular Lagrangian theory on a momentum phase
space. This facts enable one to apply Noether's first theorem both
to Lagrangian and Hamiltonian mechanics. By virtue of Noether's
first theorem, any symmetry defines a symmetry current which is an
integral of motion in Lagrangian and Hamiltonian mechanics. The
converse is not true in Lagrangian mechanics where integrals of
motion need not come from symmetries. We show that, in Hamiltonian
mechanics, any integral of motion is a symmetry current. In
particular, an energy function relative to a reference frame is a
symmetry current along a connection on a configuration bundle
which is this reference frame. An example of the global Kepler
problem is analyzed in detail.
\end{abstract}

\tableofcontents

\section{Introduction}

Noether's theorems are well known to treat symmetries of
Lagrangian systems. Noether's first theorem associates to a
Lagrangian symmetry the conserved current whose total differential
vanishes on-shell. We refer the reader to the brilliant book of
Yvette Kosmann-Schwarzbach \cite{KS} for the history and
references on the subject.

By mechanics throughout the work is meant classical non-autonomous
non-relativistic mechanics subject to time-dependent coordinate
and reference frame transformations. This mechanics is formulated
adequately as Lagrangian and Hamiltonian theory on fibre bundles
$Q\to\mathbb R$ over the time axis $\mathbb R$
\cite{book10,book98,sard98,sard132}.

Since equations of motion of mechanics almost always are of first
and second order, we restrict our consideration to first order
Lagrangian and Hamiltonian theory. Its velocity space is the first
order jet manifold $J^1Q$ of sections of a configuration bundle
$Q\to \mathbb R$, and its phase space is the vertical cotangent
bundle $V^*Q$ of $Q\to\mathbb R$.

This formulation of mechanics is similar to that of classical
field theory on fibre bundles over a smooth manifold $X$ of
dimension $n>1$ \cite{book09,sard08,book13,sard15}. A difference
between mechanics and field theory however lies in the fact that
fibre bundles over $\mathbb R$ always are trivial, and that all
connections on these fibre bundles are flat. Consequently, they
are not dynamic variables, but characterize non-relativistic
reference frames (Definition \ref{gn10}).

In Lagrangian mechanics, Noether's first theorem (Theorem
\ref{j22}) is formulated as a straightforward corollary of the
global variational formula (\ref{+421}). It associates to any
classical Lagrangian symmetry (Definition \ref{nn108}) the
conserved current (\ref{am225'}) whose total differential vanishes
on-shell.

In particular, an energy function relative to a reference frame is
the symmetry current (\ref{m228}) along a connection $\G$ on a
configuration bundle $Q\to \mathbb R$ which characterizes this
reference frame (Definition \ref{gn10}).

A key point is that, in Lagrangian mechanics, any conserved
current is an integral of motion (Theorem \ref{035'}), but the
converse need not be true (e.g., the Rung--Lenz vector (\ref{050})
in a Lagrangian Kepler model).

Hamiltonian formulation of non-autonomous non-relativistic
mechanics is similar to covariant Hamiltonian field theory on
fibre bundles \cite{book09,sard15,book16} in the particular case
of fibre bundles over $\mathbb R$ \cite{book10,sard98,book15}. In
accordance with the Legendre map (\ref{a303}) and the homogeneous
Legendre map (\ref{N41'}), a phase space and a homogeneous phase
space of mechanics on a configuration bundle $Q\to \mathbb R$ are
the vertical cotangent bundle $V^*Q$ and the cotangent bundle
$T^*Q$ of $Q$, respectively.

It should be emphasized that this is not the most general case of
a phase space of non-autonomous non-relativistic mechanics which
is defined as a fibred manifold $\Pi\to\mathbb R$ provided with a
Poisson structure such that the corresponding symplectic foliation
belongs to the fibration $\Pi\to \mathbb R$ \cite{hamo}. Putting
$\Pi=V^*Q$, we in fact restrict our consideration to Hamiltonian
systems which admit the Lagrangian counterparts on a configuration
space $Q$.

A key point is that a non-autonomous Hamiltonian system of $k$
degrees of freedom on a phase space $V^*Q$ is equivalent both to
some autonomous symplectic Hamiltonian system of $k+1$ degrees of
freedom on a homogeneous phase space $T^*Q$ (Theorem \ref{09121})
and a particular first order Lagrangian system with the
characteristic Lagrangian (\ref{Q33}) on $V^*Q$  as a
configuration space.

This facts enable one to apply Noether's first theorem both to
study symmetries in Hamiltonian mechanics (Section 7). In
particular, we show that, since Hamiltonian symmetries are vector
fields on a phase space $V^*Q$ (Definition \ref{nn234}), any
integral of motion in Hamiltonian mechanics (Definition
\ref{nn232}) is some conserved symmetry current (Theorem
\ref{0150'}).

Therefore, it may happen that symmetries and the corresponding
integrals of motion define a Hamiltonian system in full. This is
the case of commutative and noncommutative completely integrable
systems (Section 8).

In Section 9, we provides the global analysis of the Kepler
problem as an example of a mechanical system which entirely is
characterized by its symmetries. It falls into two distinct global
noncommutative completely integrable systems on different open
subsets of a phase space. Their integrals of motion form the Lie
algebras $so(3)$ and $so(2,1)$ with compact and noncompact
invariant submanifolds, respectively
\cite{book10,ijgmmp09a,book15}.

\section{Geometry of fibre bundles over $\mathbb R$}

This Section summarizes peculiarities of geometry of fibre bundles
over $\mathbb R$ \cite{book10,book98}.

Let
\beq
\pi:Q\to \mathbb R \label{gm360}
\eeq
be a fibred manifold whose base is regarded as the time axis
$\mathbb R$ parameterized by the Cartesian coordinate $t$ with
transition functions $t'=t+$const. Relative to the Cartesian
coordinate $t$, the time axis $\mathbb R$ is provided with the
global standard vector field  $\dr_t$ and the global standard
one-form $dt$ which also is a global volume form on $\mathbb R$.
The symbol $dt$ also stands for any pull-back of the standard
one-form $dt$ onto a fibre bundle over $\mathbb R$.

\begin{remark}  \label{010}
Point out one-to-one correspondence between the vector fields
$f\dr_t$, the densities $fdt$ and the real functions $f$ on
$\mathbb R$. Roughly speaking, we can neglect the contribution of
$T\mathbb R$ and $T^*\mathbb R$ to some expressions. In
particular, there is the canonical imbedding (\ref{z260}) of
$J^1Q$.
\end{remark}

In order that the dynamics of a mechanical system can be defined
at any instant $t\in\mathbb R$, we further assume that a fibred
manifold $Q\to \mathbb R$ is a fibre bundle with a typical fibre
$M$.

\begin{remark} \label{047}
A fibred manifold $Q\to\mathbb R$ is a fibre bundle if and only if
it admits an Ehresmann connection $\G$, i.e., the horizontal lift
$\G\dr_t$ onto $Q$ of the standard vector field $\dr_t$ on
$\mathbb R$ is complete.
\end{remark}

Given bundle coordinates $(t,q^i)$ on the fibre bundle
$Q\to\mathbb R$ (\ref{gm360}), the first order jet manifold $J^1Q$
of $Q\to\mathbb R$  is provided with the adapted coordinates
$(t,q^i,q^i_t)$ possessing transition functions
\mar{nn134}\beq
 t'=t+\mathrm{const.}, \qquad q'^i=q'^i(t,q^j),
 \qquad q'^i_t=(\dr_t + q^j_t\dr_j)q'^i. \label{nn134}
\eeq
In mechanics on a configuration space $Q\to\mathbb R$, the jet
manifold $J^1Q$ plays the role of a velocity space.

\begin{remark} \label{nn133} \mar{nn133}
Any fibre bundle $Q\to \mathbb R$ is trivial. Its different
trivializations
\beq
\psi: Q=  \mathbb R\times M \label{gm219}
\eeq
differ from each other in fibrations $Q\to M$. Given the
trivialization  (\ref{gm219}) coordinated by $(t,\wt q^i)$, there
is a canonical isomorphism
\beq
J^1(\mathbb R\times M)=\mathbb R\times TM, \qquad \wt q^i_t=
\dot{\wt q}^i, \label{gm220}
\eeq
that one can justify by inspection of transition functions of
coordinates $\wt q^i_t$ and  $\dot{\wt q}^i$ when transition
functions of $q^i$ are time-independent. Due to the isomorphism
(\ref{gm220}), every trivialization (\ref{gm219}) yields the
corresponding trivialization of the jet manifold
\beq
J^1Q= \mathbb R\times TM. \label{jp2}
\eeq
As a palliative variant, one develops non-relativistic mechanics
on the configuration space (\ref{gm219}) and the velocity space
(\ref{jp2}) \cite{eche,leon}. Its phase space $\mathbb R\times
T^*M$ is provided with the presymplectic form
\mar{100'}\beq
\pr^*_2\Om_T =dp_i\w dq^i \label{100'}
\eeq
which is the pull-back of the canonical symplectic form $\Om_T$
(\ref{m83}) on $T^*M$. A problem is that the presymplectic form
(\ref{100'}) is broken by time-dependent transformations.
\end{remark}

With respect to the bundle coordinates (\ref{nn134}), the
canonical imbedding of $J^1Q$ takes a form
\ben
&& \la_{(1)}: J^1Q\ni (t,q^i,q^i_t)\to (t,q^i,\dot t=1, \dot q^i=q^i_t) \in TQ, \label{z260}\\
&& \la_{(1)}=d_t=\dr_t +q^i_t\dr_i. \label{z260'}
\een
From now on, the jet manifold $J^1Q$ is identified with its image
in $TQ$ which is an affine subbundle of $TX$ modelled over the
vertical tangent bundle $VQ$ of a fibre bundle $Q\to\mathbb R$.
Using the morphism (\ref{z260}), one can define the contraction
\be
J^1Q\op\times_Q T^*Q \op\to_Q Q\times \mathbb R,  \qquad (q^i_t;
\dot{\mathrm{t}}, \dot q_i) \to
\la_{(1)}\rfloor(\dot{\mathrm{t}}dt + \dot q_idq^i) =
\dot{\mathrm{t}} + q^i_t\dot q_i,
\ee
where $(t,q^i,\dot{\mathrm{t}}, \dot q_i)$ are holonomic
coordinates on the cotangent bundle $T^*Q$.

In view of the morphism $\la_{(1)}$ (\ref{z260}), any connection
\beq
\G=dt\ot (\dr_t +\G^i\dr_i) \label{z270}
\eeq
on a fibre bundle $Q\to\mathbb R$  can be identified with a
nowhere vanishing horizontal vector field
\mar{a1.10}\beq
\G = \dr_t + \G^i \dr_i \label{a1.10}
\eeq
on $Q$ which is the horizontal lift $\G\dr_t$ of the standard
vector field $\dr_t$ on $\mathbb R$ by means of the connection
(\ref{z270}). Conversely, any vector field $\G$ on $Q$ such that
$dt\rfloor\G =1$ defines a connection on $Q\to\mathbb R$.
Therefore, the connections (\ref{z270}) further are identified
with the vector fields (\ref{a1.10}). The integral curves of the
vector field (\ref{a1.10}) coincide with the integral sections for
the connection (\ref{z270}).

Connections on a fibre bundle $Q\to\mathbb R$ constitute an affine
space modelled over a vector space of vertical vector fields on
$Q\to\mathbb R$. Accordingly, the covariant differential,
associated to a connection $\G$ on $Q\to\mathbb R$, takes its
values into the vertical tangent bundle $VQ$ of $Q\to\mathbb R$:
\beq
D_\G: J^1Q\op\to_Q VQ, \qquad \ol q^i\circ D_\G =q^i_t-\G^i.
\label{z279}
\eeq

A connection $\G$ on a fibre bundle $Q\to\mathbb R$ is obviously
flat.

\begin{theorem} \label{gn1} \mar{gn1}
Being a flat, every connection $\G$ on a fibre bundle $Q\to\mathbb
R$ defines an atlas of local constant trivializations of
$Q\to\mathbb R$ such that the associated bundle coordinates
$(t,q^i_\G)$ on $Q$ possess transition functions independent of
$t$, and
\beq
\G=\dr_t \label{z271}
\eeq
with respect to these coordinates. Conversely, every atlas of
local constant trivializations of a fibre bundle $Q\to\mathbb R$
determines a connection on  $Q\to\mathbb R$ which is equal to
(\ref{z271}) relative to this atlas.
\end{theorem}

A connection $\G$ on a fibre bundle $Q\to \mathbb R$ is said to be
complete  if the horizontal vector field (\ref{a1.10}) is
complete. By the well known theorem, a connection on a fibre
bundle $Q\to \mathbb R$ is complete if and only if it is an
Ehresmann connection. The following holds \cite{book98}.

\begin{theorem}\label{compl}
Every trivialization of a fibre bundle $Q\to \mathbb R$ yields a
complete connection on this fibre bundle. Conversely, every
complete connection $\G$ on $Q\to\mathbb R$ defines its
trivialization (\ref{gm219}) such that the horizontal vector field
(\ref{a1.10}) equals $\dr_t$ relative to the bundle coordinates
associated to this trivialization.
\end{theorem}

It follows from Theorem \ref{gn1} that, in mechanics unlike field
theory, connections $\G$ (\ref{a1.10}) on a configuration bundle
(\ref{gm360}) fail to be dynamic variables. They characterize
reference frames as follows.

From the physical viewpoint, a reference frame in mechanics
determines a tangent vector at each point of a configuration space
$Q$, which characterizes the velocity of an observer at this
point. This speculation leads to the following mathematical
definition of a reference frame in mechanics
\cite{book10,book98,sard98}.

\begin{definition}\label{gn10} \mar{gn10}
In non-relativistic mechanics, a reference frame  is the
connection $\G$ (\ref{a1.10}) on a configuration bundle
$Q\to\mathbb R$, i.e., a section of the velocity bundle $J^1Q\to
Q$.
\end{definition}

By virtue of this definition, one can think of the horizontal
vector field (\ref{a1.10}) associated to a connection $\G$ on
$Q\to\mathbb R$ as being a family of observers, while the
corresponding covariant  differential (\ref{z279}):
\mar{nn161}\beq
\ol q^i_\G= D_\G(q^i_t)= q^i_t-\G^i, \label{nn161}
\eeq
determines the relative velocity  with  respect to a reference
frame $\G$. Accordingly, $q^i_t$ are regarded as the absolute
velocities.

In accordance with Theorem \ref{gn1}, any reference frame $\G$ on
a configuration bundle $Q\to\mathbb R$ is associated to an atlas
of local constant trivializations, and vice versa. A connection
$\G$ takes the form $\G=\dr_t$ (\ref{z271}) with respect to the
corresponding coordinates $(t,q^i_\G)$, whose transition functions
are independent of time. One can think of these coordinates as
also being a reference frame, corresponding to the connection
(\ref{z271}). They are called the adapted coordinates  to a
reference frame $\G$. Thus, we come to the following definition,
equivalent to Definition \ref{gn10}.

\begin{definition}\label{gn11}
In mechanics, a reference frame is an atlas of local constant
trivializations of a configuration bundle $Q\to\mathbb R$.
\end{definition}

In particular, with respect to the coordinates $q^i_\G$ adapted to
a reference frame $\G$, the velocities relative to this reference
frame (\ref{nn161}) coincide with the absolute ones
\be
\rrq^i_\G=D_\G(q^i_{\G t})=q^i_{\G t}.
\ee

A reference frame is said to be complete if the associated
connection $\G$ is complete. By virtue of Theorem \ref{compl},
every complete reference frame defines a trivialization  of a
bundle $Q\to\mathbb R$, and vice versa.

\section{Lagrangian mechanics. Integrals of motion}

As was mentioned above, our exposition is restricted to first
order Lagrangian theory on a fibre bundle $Q\to \mathbb R$
\cite{book10,book98,sard132}. This is a standard case of
Lagrangian mechanics.

In mechanics, a first order Lagrangian is defined as a horizontal
density
\mar{23f2'}\beq
L=\cL dt, \qquad \cL: J^1Q\to \mathbb R, \label{23f2'}
\eeq
on a velocity space $J^1Q$.

The corresponding second order Euler--Lagrange operator, termed
the Lagrange operator,  reads
\mar{305'}\ben
&&\dl L= (\dr_i\cL- d_t\dr^t_i\cL) \thh^i\w dt, \label{305'}\\
&& d_t=\dr_t +q^i_t\dr_i + q^i_{tt}\dr_i^t. \nonumber
\een
Let us further use the notation
\mar{03}\beq
\pi_i=\dr^t_i\cL, \qquad \pi_{ji}=\dr_j^t\dr_i^t\cL. \label{03}
\eeq

The kernel $\gE_L=\Ker\dl L\subset J^2Q$ of the Lagrange operator
(\ref{305'}) defines a  second order Lagrange equation
\mar{b327'}\beq
(\dr_i- d_t\dr^t_i)\cL=0 \label{b327'}
\eeq
on $Q$. Its classical solutions  are (local) sections $c$ of a
fibre bundle $Q\to\mathbb R$ whose second order jet prolongations
$J^2c=\dr_{tt}c$ live in $\gE_L$ (\ref{b327'}).

Every first order Lagrangian $L$ (\ref{23f2'}) yields the Legendre
map
\mar{a303}\beq
\wh L:J^1Q\ar_Q V^*Q,\qquad  p_i \circ\wh L = \pi_i, \label{a303}
\eeq
where the Legendre bundle $\Pi=V^*Q$ is the vertical cotangent
bundle $V^*Q$ of $Q\to\mathbb R$ provided with holonomic
coordinates $(t,q^i,p_i)$. As was mentioned above, it plays the
role of a phase space of mechanics on a configuration space
$Q\to\mathbb R$. The corresponding Lagrangian constraint space is
\beq
N_L=\wh L(J^1Q)\label{jkl}
\eeq

\begin{definition} \label{d11} \mar{d11}
A Lagrangian $L$ is said to be:

$\bullet$ hyperregular if the Legendre map $\wh L$ is a
diffeomorphism;

$\bullet$ regular if $\wh L$ is a local diffeomorphism, i.e.,
$\det(\pi_{ij})\neq 0$;

$\bullet$ almost regular if the Lagrangian constraint space $N_L$
(\ref{jkl}) is a closed imbedded subbundle $i_N:N_L\to V^*Q$ of
the Legendre bundle $V^*Q\to Q$ and the Legendre map
\mar{cmp12}\beq
\wh L:J^1Q\to N_L \label{cmp12}
\eeq
is a fibred manifold with connected fibres.
\end{definition}

Given a first order Lagrangian $L$, there is the global
decomposition (called the variational formula)
\mar{+421}\beq
dL=\dl L-d_t\Xi_L \label{+421}
\eeq
where the Lepage equivalent $\Xi_L$ of $L$ is the
Poincar\'e--Cartan form
\mar{303'}\beq
H_L=\pi_i dq^i -(\pi_iq^i_t-\cL)dt \label{303'}
\eeq
(see the notation (\ref{03})).  This form takes its values into a
subbundle
\be
J^1Q\op\times_Q T^*Q\subset T^*J^1Q.
\ee
Hence, it defines the homogeneous Legendre map
\mar{N41'}\beq
\wh H_L: J^1Q\to Z_Y=T^*Q, \label{N41'}
\eeq
whose range is an imbedded subbundle
\beq
Z_L= \wh H_L(J^1Q)\subset T^*Q \label{23f10'}
\eeq
of the homogeneous Legendre bundle $Z_Y=T^*Q$ (\ref{N41'}). Let
$(t,q^i,p_0=\dot{\mathrm{t}},p_i=\dot q_i)$ denote holonomic
coordinates on $T^*Q$ which possess transition functions
\beq
{p'}_i = \frac{\dr q^j}{\dr{q'}^i}p_j, \qquad p'_0=
\left(p_0+\frac{\dr q^j}{\dr t}p_j\right). \label{2.3a}
\eeq
With respect to these coordinates, the homogeneous Legendre map
$\wh H_L$ (\ref{N41'}) reads
\be
(p_0,p_i)\circ \wh H_L =(\cL-q^i_t\p_i, \p_i).
\ee

In view of the morphism (\ref{N41'}), the cotangent bundle $T^*Q$
plays the role of a homogeneous phase space of mechanics.

A glance at the transition functions (\ref{2.3a}) shows that the
canonical map
\mar{b418a}\beq
\zeta:T^*Q\to V^*Q, \label{b418a}
\eeq
ia a one-dimensional affine bundle over the vertical cotangent
bundle $V^*Q$. Herewith, the Legendre map $\wh L$ (\ref{a303}) is
exactly the composition of morphisms
\be
\wh L=\zeta\circ H_L:J^1Q \op\to_Q V^*Q.
\ee

\begin{remark} \label{nn145} \mar{nn145}
The Poincar\'e--Cartan form $H_L$ (\ref{303'}) also is the
Poincar\'e--Cartan form $H_L=H_{\wt L}$ of a first order
Lagrangian
\beq
\wt L=\wh h_0(H_L) = (\cL + (q_{(t)}^i - q_t^i)\p_i)dt, \qquad \wh
h_0(dq^i)=q^i_{(t)} dt,\label{cmp80a}
\eeq
on the  repeated jet manifold $J^1J^1Q$. The Lagrange operator
(\ref{305'}) for $\wt L$ (\ref{cmp80a}) reads
\be
&& \dl\wt L  = [(\dr_i\cL - \wh d_t\p_i + \dr_i\p_j(q_{(t)}^j -
q_t^j))dq^i + \dr_i^t\p_j(q_{(t)}^j - q_t^j) dq_t^i]\w dt, \\
&& \wh d_t=\dr_t +\wh q^i_t\dr_i + q^i_{tt}\dr^t_i. \nonumber
\ee
Its kernel $\Ker\dl \ol L\subset J^1J^1Q$ defines the Cartan
equation
\mar{b336c}\beq
\dr_i^t\p_j(q_{(t)}^j - q_t^j)=0, \qquad \dr_i\cL - \wh d_t\p_i +
\dr_i\p_j(q_{(t)}^j - q_t^j)=0 \label{b336c}
\eeq
on a velocity space $J^1Q$.
\end{remark}

In mechanics, the Lagrange equation (\ref{b327'}) as like as the
Hamilton one (\ref{z20a}) is an ordinary differential equation.
One can think of its classical solutions $s(t)$ as being a motion
in a configuration space $Q$. In this case, the notion of
integrals of motion can be introduced as follows
\cite{book10,book15}.

In a general setting, let an equation of motion of a mechanical
system is an $r$-order differential equation $\gE$ on a fibre
bundle $Y\to\mathbb R$ given by a closed subbundle of the jet
bundle $J^rY\to\mathbb R$.

\begin{definition} \label{026} \mar{026}
An integral of motion of this mechanical system is defined as a
$(k<r)$-order differential operator $\Phi$ on $Y$ such that $\gE$
belongs to the kernel of an $r$-order jet prolongation of a
differential operator $d_t\Phi$, i.e.,
\mar{021}\ben
&& J^{r-k-1}(d_t\Phi)|_{\gE}=J^{r-k}\Phi|_{\gE}=0, \label{021} \\
&&  d_t=\dr_t +y^a_t\dr_a + y^a_{tt}\dr^t_a + \cdots. \nonumber
\een
\end{definition}

It follows that an integral of motion $\Phi$ is constant on
classical solutions $s$ of a differential equation $\gE$, i.e.,
there is the differential conservation law
\beq
(J^ks)^*\Phi=\mathrm{const}., \qquad (J^{k+1}s)^*d_t\Phi=0.
\label{020}
\eeq

We agree to write the condition (\ref{021}) as a weak equality
\beq
J^{r-k-1}(d_t\Phi)\ap 0, \label{022}
\eeq
which holds on-shell, i.e., on solutions of a differential
equation $\gE$ by the formula (\ref{020}).

In mechanics, we restrict our consideration to integrals of motion
$\Phi$ which are functions on $J^kY$. As was mentioned above,
equations of motion of mechanics mainly are either of first or
second order. Accordingly, their integrals of motion are functions
on $Y=J^0Y$ or $J^1Y$. In this case, the corresponding weak
equality (\ref{021}) takes a form
\mar{027}\beq
d_t\Phi\ap 0 \label{027}
\eeq
of a weak conservation law of an integral of motion.

Integrals of motion can come from symmetries. This is the case
both of Lagrangian mechanics on a configuration space $Y=Q$
(Theorems \ref{035'} -- \ref{nn191}) and Hamiltonian mechanics on
a phase space $Y=V^*Q$ (Theorem \ref{nn192}).

\begin{definition} \label{025} \mar{025}
Let an equation of motion of a mechanical system be an $r$-order
differential equation $\gE\subset J^rY$. Its infinitesimal
symmetry (or, shortly, a symmetry) is defined as a vector field on
$J^rY$ whose restriction to $\gE$ is tangent to $\gE$.
\end{definition}

Following Definition \ref{025}, let us introduce a notion of the
symmetry of differential operators in the following relevant case.
Let us consider an $r$-order differential operator on a fibre
bundle $Y\to\mathbb R$ which is represented by an exterior form
$\cE$ on $J^rY$. Let its kernel $\Ker\cE$ be an $r$-order
differential equation on $Y\to\mathbb R$.

\begin{theorem} \label{082} \mar{082}
It is readily justified that a vector field $\vt$ on $J^rY$ is a
symmetry of the equation $\Ker\cE$ in accordance with Definition
\ref{025} if and only if $\bL_\vt \cE\ap 0$.
\end{theorem}

Motivated by Theorem \ref{082}, we come to the following notion.

\begin{definition} \label{084}  \mar{084} Let $\cE$ be the above
mentioned differential operator. A vector field $\vt$ on $J^rY$ is
termed the symmetry of a differential operator $\cE$ if the Lie
derivative $\bL_\vt \cE$ vanishes.
\end{definition}

By virtue of Theorem \ref{082}, a symmetry of a differential
operator $\cE$ also is a symmetry of a differential equation
$\Ker\cE$.

Note that there exist integrals of motion which are not associated
with symmetries of an equation of motion, e.g., the Rug--Lenz
vector (\ref{050}) in a Lagrangian Kepler system (Section 9).

\section{Noether's first theorem: Energy conservation laws}

In Lagrangian mechanics, integrals of motion come from symmetries
of a Lagrangian (Theorem \ref{035'}) in accordance with Noether's
first theorem (Theorem \ref{j22}). However as was mentioned above,
not all integrals of motion are of this type.

In the framework of first order Lagrangian formalism, we restrict
our consideration to classical symmetries  of a Lagrangian system
represented by vector fields $\up$ on a configuration bundle $Q\to
\mathbb R$. Moreover, not concerned with time-reparametrization,
we deal with vector fields
\mar{j15'}\beq
\up=\up^t\dr_t +\up^i(t,q^i)\dr_i, \qquad \up^t=0,1. \label{j15'}
\eeq
Their first order jet prolongation onto the velocity space $J^1Q$
read
\mar{a23f41}\beq
J^1\up= \up^t\dr_t + \up^i\dr_i + d_t \up^i\dr^t_i. \label{a23f41}
\eeq

Let $L$ be the Lagrangian (\ref{23f2'}) on a velocity space
$J^1Q$. Due ti the variational formula (\ref{+421}), its Lie
derivative $\bL_{J^1\up} L$ along the $J^1\up$ (\ref{a23f41})
obeys the relation (called the first variational formula)
\mar{a23f42}\beq
\bL_{J^1\up}L= \up_V\rfloor\dl L + d_H(\up\rfloor H_L),
\label{a23f42}
\eeq
where $H_L$ is the Poincar\'e--Cartan form (\ref{303'}). Its
coordinate expression reads
\mar{J4'}\beq
[\up^t\dr_t+ \up^i\dr_i +d_t \up^i\dr^t_i]\cL = (\up^i-q^i_t
\up^t)\cE_i + d_t[\pi_i(\up^i-\up^t q^i_t) +\up^t\cL]. \label{J4'}
\eeq

\begin{definition} \label{nn108} \mar{nn108}
The vector field $\up$ (\ref{j15'}) on $Q$ is called the
Lagrangian symmetry (or, shortly, the symmetry) of a Lagrangian
$L$ if the Lie derivative $\bL_{J^1\up} L$ of $L$ is $d_t$-exact,
i.e.,
\mar{22f1}\beq
\bL_{J^1 \up} L=d_t\si dt \label{22f1}
\eeq
where $\si$ is a function on $J^1Q$.
\end{definition}

Then Noether's first theorem is formulated as follows.

\begin{theorem} \label{j22} \mar{j22} If
the vector field $\up$ (\ref{j15'}) is a symmetry of a Lagrangian
$L$, a corollary of the first variational formula (\ref{J4'})
on-shell is the weak Lagrangian conservation law
\mar{d22f2}\beq
0\ap d_t(\up\rfloor H_L -\si)dt\ap  d_t(\pi_i(\up^i- \up^t q^i_t)
+\up^t\cL -\si)dt\ap -d_t\cJ_\up dt \label{d22f2}
\eeq
of a symmetry current
\mar{am225'}\beq
\cJ_\up=-(\up\rfloor H_L-\si) =-(\pi_i(\up^i-\up^tq^i_t) +
\up^t\cL-\si) \label{am225'}
\eeq
along $\up$. The symmetry current (\ref{am225'}) obviously is
defined with the accuracy to a constant summand.
\end{theorem}

\begin{theorem} \label{035'} \mar{035'}
It is readily observed that the conserved current $\cJ_\up$
(\ref{am225'}) along a classical symmetry is a function on a
velocity space $J^1Q$, and it is an  integral of motion of a
Lagrangian system in accordance with Definition \ref{026}.
\end{theorem}

\begin{theorem} \label{nn191} \mar{nn191}
If a symmetry $\up$ of a Lagrangian $L$ is classical, this is a
symmetry of the Lagrange operator $\dl L$ (\ref{305'}) and, as a
consequence, an infinitesimal symmetry of the Lagrange equation
$\gE_L$ (\ref{b327'}) (Theorem \ref{082}).
\end{theorem}

\begin{remark} \label{045'}  \mar{045'}Given a Lagrangian $L$, let
$\wh L$ be its partner (\ref{cmp80a}) on the repeated jet manifold
$J^1J^1Q$. Since $H_L$ (\ref{a303}) is the Poincar\'e--Cartan form
both for $L$ and $\wt L$, a Lagrangian $\wh L$ does not lead to
new conserved currents.
\end{remark}

If a symmetry $\up$ of a Lagrangian $L$ is exact, i.e,
\be
\bL_{J^1\up}L=0,
\ee
the first variational formula (\ref{a23f42}) takes a form
\mar{nn230}\beq
 0= \up_V\rfloor\dl L + d_H(\up\rfloor H_L). \label{nn230}
\eeq
It leads to the weak conservation law (\ref{d22f2}):
\beq
0\ap -d_t\cJ_\up, \label{gm488}
\eeq
of the symmetry current
\mar{m225}\beq
\cJ_\up=-\up\rfloor H_L=-\pi_i(\up^i-\up^tq^i_t) - \up^t\cL
\label{m225}
\eeq
along a vector field $\up$.

\begin{remark}
The first variational formula  (\ref{nn230}) also can be utilized
when a Lagrangian possesses exact symmetries, but an equation of
motion is a sum
\mar{z0122}\beq
(\dr_i- d_t\dr^t_i)\cL +f_i(t,q^j,q^j_t)=0 \label{z0122}
\eeq
of a Lagrange equation and an additional non-Lagrangian external
force. Let us substitute $\cE_i=-f_i$ from this equality in the
first variational formula (\ref{nn230}). Then we have the weak
transformation law
\be
(\up^i-q^i_t\up^t)f_i\ap d_t\cJ_\up
\ee
of the symmetry current $\cJ_\up$ (\ref{m225}) on the shell
(\ref{z0122}).
\end{remark}

It is readily observed that the first variational formula
(\ref{J4'}) is linear in a vector field $\up$. Therefore, one can
consider superposition of the equalities (\ref{J4'}) for different
vector fields.

For instance, if $\up$ and $\up'$ are projectable vector fields
(\ref{j15'}), they are projected onto the standard vector field
$\dr_t$ on $\mathbb R$, and the difference of the corresponding
equalities (\ref{J4'}) results in the first variational formula
(\ref{J4'}) for a vertical vector field $\up-\up'$.

Conversely, every vector field $\up$ (\ref{j15'}), projected onto
$\dr_t$, can be written as a sum
\mar{gm490}\beq
\up=\G +v \label{gm490}
\eeq
of some reference frame (\ref{a1.10}):
\mar{gm489}\beq
\G=\dr_t +\G^i\dr_i, \label{gm489}
\eeq
and a vertical vector field $v$ on $Q\to\mathbb R$.

It follows that the first variational formula (\ref{J4'}) for the
vector field $\up$ (\ref{j15'}) can be represented as a
superposition of those for the reference frame $\G$ (\ref{gm489})
and a vertical vector field $v$.

If $\up=v$ is a vertical vector field, the first variational
formula (\ref{J4'}) reads
\be
(v^i\dr_i +d_t v^i \dr^t_i)\cL = v^i\cE_i + d_t(\pi_i v^i).
\ee
If $v$ is an exact symmetry of $L$, we obtain from (\ref{gm488})
the Noether conservation law
\be
0\ap d_t(\pi_i v^i)
\ee
of the  Noether current
\mar{z384}\beq
\cJ_v=-\pi_i v^i, \label{z384}
\eeq
which is a Lagrangian integral of motion by virtue of Theorem
\ref{035'}.

\begin{remark} \label{057} \mar{057}
Let us assume that, given a trivialization $Q= \mathbb R\times M$
in bundle coordinates $(t,q^i)$, a Lagrangian $L$ is independent
of some coordinate $q^a$. Then a vertical vector field $v=\dr_i$
is an exact symmetry of $L$, and we have the conserved Noether
current $\cJ_v=-\pi_i$ (\ref{z384}) which is an integral of
motion.
\end{remark}

In the case of the reference frame $\G$ (\ref{gm489}), where
$\up^t=1$, the first variational formula (\ref{J4'}) reads
\mar{m227}\beq
(\dr_t +\G^i\dr_i +d_t\G^i\dr_i^t)\cL = (\G^i-q^i_t)\cE_i -
d_t(\pi_i(q^i_t-\G^i) -\cL), \label{m227}
\eeq
where
\mar{m228}\beq
E_\G=\cJ_\G= \pi_i(q^i_t -\G^i) -\cL \label{m228}
\eeq
is called the energy function relative to a reference frame $\G$
\cite{book10,sard98}.

With respect to the coordinates $q^i_\G$ adapted to a reference
frame $\G$, the first variational formula (\ref{m227}) takes a
form
\beq
\dr_t\cL = -q^i_{\G t}\cE_i - d_t(\pi_iq^i_{\G t}
-\cL),\label{m229}
\eeq
and the $E_\G$ (\ref{m228}) coincides with the canonical energy
function
\mar{nn162}\beq
E_L=\pi_iq^i_{\G t} -\cL. \label{nn162}
\eeq
A glance at the expression (\ref{m229}) shows that the vector
field $\G$ (\ref{gm489}) is an exact symmetry of a Lagrangian $L$
if and only if, written with respect to coordinates adapted to
$\G$, this Lagrangian is independent on the time $t$. In this
case, the energy function $E_\G$ (\ref{m229}) relative to a
reference frame $\G$ is conserved:
\mar{nn163}\beq
0\ap -d_t E_\G. \label{nn163}
\eeq
It is a Lagrangian integral of motion in accordance with Theorem
\ref{035'}.

Since any vector field $\up$ (\ref{j15'}) can be represented as
the sum (\ref{gm490}) of the reference frame $\G$ (\ref{gm489})
and a vertical vector field $v$, the symmetry current (\ref{m225})
along the vector field $\up$ (\ref{j15'}) is a sum
\be
\cJ_\up=E_\G +\cJ_v
\ee
of the  Noether current $\cJ_v$ (\ref{z384}) along a vertical
vector field $v$ and the energy function $E_\G$ (\ref{m228})
relative to a reference frame $\G$. Conversely, energy functions
relative to different reference frames $\G$ and $\G'$ differ from
each other in the Noether current (\ref{z384}) along a vertical
vector field $\G'-\G$:
\mar{nn170}\beq
E_\G-E_{\G'}=\pi_i(\G'^i-\G^i)=\cJ_{\G-\G'}. \label{nn170}
\eeq

\begin{example} \label{nn160} \mar{nn160}
Given a configuration space $Q$ of a mechanical system and the
connection $\G$ (\ref{gm489}) on $Q\to\mathbb R$, let us consider
a quadratic Lagrangian
\mar{cqg48}\beq
L=\frac12m_{ij}(t,q^k)( q^i_t-\G^i)(q^j_t-\G^j)dt, \label{cqg48}
\eeq
where $m_{ij}$ is a non-degenerate positive-definite fibre metric
in the vertical tangent bundle $VQ\to Q$. It is called the mass
tensor.  Such a Lagrangian is globally defined owing to the linear
transformation laws of the relative velocities $\ol q^i_\G$
(\ref{nn161}). Let $q^i_\G$ be fibre coordinates adapted to a
reference frame $\G$. Then the Lagrangian (\ref{cqg48}) reads
\mar{jp71}\beq
L=\frac12\ol m_{ij}(q^k)q^i_{\G t} q^j_{\G t} dt. \label{jp71}
\eeq
Since coordinates $q^i_\G$ possess time-independent transition
functions, let us assume that a mass tensor $\ol m_{ij}$ is
independent of time. It is readily observed that, in this case, a
horizontal vector field $\G\dr_t=\G=\dr_t$ is an exact symmetry of
the Lagrangian $L$ (\ref{jp71}) that leads to a weak conservation
law of the canonical energy function (\ref{nn162}):
\mar{nn164}\beq
E_L=\frac12\ol m_{ij}(q^k)q^i_{\G t} q^j_{\G t} dt. \label{nn164}
\eeq
With respect to arbitrary bundle coordinates $(t,q^i)$ on $Q$, the
energy function (\ref{nn164}) takes a form
\be
E_\G=\frac12m_{ij}(t,q^k)( q^i_t-\G^i)(q^j_t-\G^j).
\ee
This is an energy function relative to a reference frame $\G$.
\end{example}

\begin{example} \label{nn167} \mar{nn167}
Let us consider a one-dimensional motion of a point mass $m_0$
subject to friction. It is described by the dynamic equation
\be
m_0q_{tt}=-kq_t, \qquad k> 0,
\ee
on a configuration space $\mathbb R^2\to\mathbb R$ coordinated by
$(t,q)$. This equation is a Lagrange equation of a Lagrangian
\be
L=\frac12 m_0\exp\left[\frac{k}{m_0}t\right]q_t^2 dt,
\ee
termed the Havas Lagrangian \cite{book10,rie}. It is readily
observed that the Lie derivative of this Lagrangian along a vector
field
\beq
\G=\dr_t- \frac12\frac{k}{m_0}q\dr_q \label{gm23}
\eeq
vanishes. Consequently, we have the conserved energy function
(\ref{m228}) with respect to the reference frame $\G$
(\ref{gm23}). This energy function reads
\be
E_\G=\frac12m_0\exp\left[\frac{k}{m_0}t\right]
q_t\left(q_t+\frac{k}{m_0}q\right).
\ee
\end{example}

In Section 9, an example of the global Kepler problem is analyzed
in details.

\section{Noether's third theorem: Gauge symmetries}

In mechanics, we follow general definition of gauge symmetries of
Lagrangian theory on fibre bundles \cite{book09,gauge09,sard14}.
It is given by a vector field
\mar{gg2'}\beq
u=\left(\op\sum_{0\leq|\La|\leq m}
u^{i\La}_a(t,q^j_\Si)\chi^a_\La\right)\dr_i. \label{gg2'}
\eeq
on a configuration space $Q$ which depends on sections $\chi$ of
some vector bundle $E\to\mathbb R$.

If a Lagrangian $L$ admits the gauge symmetry $u$ (\ref{gg2'}),
the weak conservation law (\ref{d22f2}) of the corresponding
symmetry current $\cJ_u$ (\ref{am225'}) holds. Because gauge
symmetries depend on derivatives of gauge parameters, all gauge
conservation laws in first order Lagrangian mechanics possess the
following peculiarity.

\begin{theorem} \label{supp2} \mar{supp2}
If $u$ (\ref{gg2'}) is a gauge symmetry of a first order
Lagrangian $L$, the corresponding symmetry current $\cJ_u$
(\ref{am225'}) vanishes on-shell, i.e., $\cJ\ap 0$.
\end{theorem}

\begin{proof}
Let a gauge symmetry $u$ be at most of jet order $N$ in gauge
parameters. Then the current $\cJ_u$ is decomposed into a sum
\mar{g2g}\ben
\cJ_u= \op\sum_{1<|\La|\leq N} J^\La_a\chi^a_\La + J^t_a\chi^a_t
+J_a\chi^a.\label{g2g}
\een
The first variational formula (\ref{J4'}) takes a form
\be
0=\left[\op\sum_{|\La|=1}^N u^i{}_a^\La\chi^a_\La
+u^i_a\chi^a\right]\cE_i - d_t\left(\op\sum_{|\La|=1}^N
J^\La_a\chi^a_\La +J_a\chi^a\right).
\ee
It falls into a set of equalities for each $\chi^a_{t\La}$,
$\chi^a_\La$, $|\La|=1,\ldots,N$, and $\chi^a$ as follows:
\mar{g4g,-7}\ben
&& 0=J^\La_a, \qquad |\La|=N, \label{g4g}\\
&& 0=-u_a^{it\La}\cE_i+ J^\La_a
+d_t J^{t\La}_a, \qquad 1\leq |\La|<N, \label{g5g}\\
&& 0= -u_a^{it}\cE_i + J_a + d_t J^t_a,
\label{g6g}\\
&& 0= -u^i_a\cE_i + d_tJ_a. \label{g7g}
\een
With the equalities (\ref{g4g}) -- (\ref{g6g}), the decomposition
(\ref{g2g}) takes a form
\be
\cJ_u= \op\sum_{1<|\La|< N}[(u_a^{it\La}\cE_i - d_t
J^{t\La}_a]\chi^a_\La+ (u_a^{itt}\cE_i -d_t J^{tt}_a)\chi^a_t +
(u_a^{it}\cE_i +- d_t J^t_a)\chi^a.
\ee
A direct computation leads to the expression
\mar{g8g}\ben
&& \cJ_u=\left(\op\sum_{1\leq |\La|<
N}u_a^{it\La}
\chi^a_\La+ u_a^{it} \chi^a\right)\cE_i -\label{g8g}\\
&& \qquad \left(\op\sum_{1\leq |\La|< N}d_t
J^{t\La}_a\chi^a_\La+ d_t J^t_a\chi^a\right). \nonumber
\een
The first summand of this expression vanishes on-shell. Its second
one contains the terms $d_t J^\La_a$, $|\La|=1,\ldots, N$. By
virtue of the equalities (\ref{g5g}), every $d_t J^\La_a$,
$|\La|<N$, is expressed in the terms vanishing on-shell and the
term $d_td_t J^{t\La}_a$. Iterating the procedure and bearing in
mind the equality (\ref{g4g}), one can easily show that the second
summand of the expression (\ref{g8g}) also vanishes on-shell.
Thus, a current $\cJ_u$ vanishes on-shell.
\end{proof}

Let us note that the statement of Theorem \ref{supp2} is a
particular case of so-called Noether's third theorem that currents
of gauge symmetries in Lagrangian theory are reduced to a
superpotential \cite{fat94,book09,got92,gauge09} because a
superpotential equals zero on $X=\mathbb R$ .

\section{Non-autonomous Hamiltonian mechanics}

As was mentioned above, a Hamiltonian formulation of
non-autonomous non-relativistic mechanics is similar to covariant
Hamiltonian field theory on fibre bundles
\cite{book09,sard15,book16} in the particular case of fibre
bundles over $\mathbb R$ \cite{book10,sard98,sard132,book15}.

In accordance with the Legendre map (\ref{a303}) and the
homogeneous Legendre map (\ref{N41'}), a phase space and a
homogeneous phase space of mechanics on a configuration bundle
$Q\to \mathbb R$ are the vertical cotangent bundle $V^*Q$ and the
cotangent bundle $T^*Q$ of $Q$, respectively.

A key point is that a non-autonomous Hamiltonian system of $k$
degrees of freedom on a phase space $V^*Q$ is equivalent both to
some autonomous symplectic Hamiltonian system of $k+1$ degrees of
freedom on a homogeneous phase space $T^*Q$ (Theorem \ref{09121})
and to a particular first order Lagrangian system with the
characteristic Lagrangian (\ref{Q33}) on a configuration space
$V^*Q$.

The cotangent bundle $T^*Q$ is endowed with holonomic coordinates
$(t,q^i,p_0,p_i)$, possessing the transition functions
(\ref{2.3a}). It admits the canonical Liouville form
(\ref{nn237}):
\mar{N43a}\beq
\Xi_T=p_0dt + p_idq^i, \label{N43a}
\eeq
the canonical symplectic form (\ref{m83}):
\mar{m91'}\beq
\Om_T=d\Xi_T=dp_0\w dt +dp_i\w dq^i, \label{m91'}
\eeq
and the corresponding canonical Poisson bracket (\ref{nn238}):
\mar{m116}\beq
\{f,g\}_T =\dr^0f\dr_tg - \dr^0g\dr_tf +\dr^if\dr_ig-\dr^ig\dr_if,
\quad f,g\in C^\infty(T^*Q). \label{m116}
\eeq

There is the canonical one-dimensional affine bundle
(\ref{b418a}):
\mar{nn169}\beq
\zeta:T^*Q\to V^*Q. \label{nn169} \eeq
 A glance at the
transformation law (\ref{2.3a}) shows that it is a trivial affine
bundle. Indeed, given a global section $h$ of $\zeta$, one can
equip $T^*Q$ with a global fibre coordinate
\mar{09151}\beq
I_0=p_0-h, \qquad I_0\circ h=0, \label{09151}
\eeq
possessing the identity transition functions. With respect to
coordinates
\mar{09150}\beq
(t,q^i,I_0,p_i), \qquad i=1,\ldots,m, \label{09150}
\eeq
the fibration (\ref{nn169}) reads
\be
\zeta: \mathbb R\times V^*Q \ni (t,q^i,I_0,p_i)\to (t,q^i,p_i)\in
V^*Q,
\ee
where $(t,q^i,p_i)$ are holonomic coordinates on the vertical
cotangent bundle $V^*Q$.

Let us consider a subring of $C^\infty(T^*Q)$ which comprises the
pull-back $\zeta^*f$ onto $T^*Q$ of functions $f$ on the vertical
cotangent bundle $V^*Q$ by the fibration $\zeta$ (\ref{nn169}).
This subring is closed under the Poisson bracket (\ref{m116}).
Then by virtue of the well known theorem \cite{book10,vais}, there
exists a degenerate coinduced Poisson bracket
\mar{m72}\beq
\{f,g\}_V = \dr^if\dr_ig-\dr^ig\dr_if, \qquad f,g\in
C^\infty(V^*Q), \label{m72}
\eeq
on  a phase space $V^*Q$ such that
\beq
\zeta^*\{f,g\}_V=\{\zeta^*f,\zeta^*g\}_T.\label{m72'}
\eeq
Holonomic coordinates on $V^*Q$ are canonical for the Poisson
structure (\ref{m72}).

With respect to the Poisson bracket (\ref{m72}), the Hamiltonian
vector fields  of functions on $V^*Q$  read
\mar{m73,094}\ben
&& \vt_f = \dr^if\dr_i- \dr_if\dr^i, \qquad f\in C^\infty(V^*Q),
\label{m73}\\
&& \{f,f'\}_V=\vt_f\rfloor df', \qquad [\vt_f,\vt_{f'}]=\vt_{\{f,f'\}_V}. \label{094}
\een
They are vertical vector fields on $V^*Q\to \mathbb R$.
Accordingly, the characteristic distribution of the Poisson
structure (\ref{m72}) is the vertical tangent bundle $VV^*Q\subset
TV^*Q$ of a fibre bundle $V^*Q\to \mathbb R$. The corresponding
symplectic foliation on a phase space $V^*Q$ coincides with the
fibration $V^*Q\to \mathbb R$.

However, the Poisson structure (\ref{m72}) fails to provide any
dynamic equation on a phase space $V^*Q\to\mathbb R$ because
Hamiltonian vector fields (\ref{m73}) of functions on $V^*Q$ are
vertical vector fields. Hamiltonian dynamics on $V^*Q$ is
described as a particular Hamiltonian dynamics on fibre bundles
\cite{book10,sard98,book15}.

A Hamiltonian  on a phase space $V^*Q\to\mathbb R$ is defined as a
global section
\mar{ws513}\beq
h:V^*Q\to T^*Q, \qquad p_0\circ h=\cH(t,q^j,p_j), \label{ws513}
\eeq
of the affine bundle $\zeta$ (\ref{nn169}). Given the Liouville
form $\Xi_T$ (\ref{N43a}) on $T^*Q$, this section yields the
pull-back Hamiltonian form
\mar{b4210}\beq
H=(-h)^*\Xi_T= p_k dq^k -\cH dt  \label{b4210}
\eeq
on $V^*Q$. This is the well-known invariant of Poincar\'e--Cartan
\cite{book10}.

It should be emphasized that, in contrast with a Hamiltonian in
autonomous mechanics, the Hamiltonian $\cH$ (\ref{ws513}) is not a
function on $V^*Q$, but it obeys the transformation law
\be
\cH'(t,q'^i,p'_i)=\cH(t,q^i,p_i)+ p'_i\dr_t q'^i.
\ee

\begin{remark} \label{ws512}
Any connection $\G$ (\ref{a1.10}) on a configuration bundle
$Q\to\mathbb R$ defines the global section $h_\G=p_i\G^i$
(\ref{ws513}) of the affine bundle $\zeta$ (\ref{nn169}) and the
corresponding Hamiltonian form
\mar{ws515}\beq
H_\G= p_k dq^k -\cH_\G dt= p_k dq^k -p_i\G^i dt. \label{ws515}
\eeq
Furthermore, given a connection $\G$, any Hamiltonian form
(\ref{b4210}) admits a splitting
\mar{m46'}\beq
H= H_\G -\cE_\G dt, \label{m46'}
\eeq
where
\mar{xx60}\beq
\cE_\G=\cH-\cH_\G=\cH- p_i\G^i \label{xx60}
\eeq
is a function on $V^*Q$. It is called the Hamiltonian function
relative to a reference frame $\G$. With respect to the
coordinates adapted to a reference frame $\G$, we have
$\cE_\G=\cH$. Given different reference frames $\G$ and $\G'$, the
decomposition (\ref{m46'}) leads at once to a relation
\be
 \cE_{\G'}=\cE_\G + \cH_\G -\cH_{\G'}=\cE_\G + (\G^i
-\G'^i)p_i
\ee
(cf. (\ref{nn170})) between the Hamiltonian functions with respect
to different reference frames.
\end{remark}

Given the Hamiltonian form $H$ (\ref{b4210}), there exists a
unique Hamiltonian connection
\mar{z3}\beq
\g_H=\dr_t + \dr^k\cH\dr_k- \dr_k\cH\dr^k. \label{z3}
\eeq
on $V^*Q\to \mathbb R$ such that $\g_H\rfloor dH=0$. It  yields a
first order dynamic Hamilton equation
\mar{z20a}\beq
q^k_t=\dr^k\cH, \qquad p_{tk}=-\dr_k\cH \label{z20a}
\eeq
on $V^*Q\to\mathbb R$, where $(t,q^k,p_k,q^k_t,p_{tk})$ are
adapted coordinates on the first order jet manifold $J^1V^*Q$ of
$V^*Q\to\mathbb R$.

A classical solution of the Hamilton equation (\ref{z20a}) is an
integral section $r$ for the connection $\g_H$ (\ref{z3}).

We agree to call $(V^*Q,H)$ the Hamiltonian system of $k=\di Q-1$
degrees of freedom.

In order to describe evolution of a Hamiltonian system at any
instant, the Hamilton connection $\g_H$ (\ref{z3}) is assumed to
be complete, i.e., it is an Ehresmann connection (Remark
\ref{047}). In this case, the Hamilton equation (\ref{z20a})
admits a unique global classical solution through each point of a
phase space $V^*Q$. By virtue of Theorem \ref{compl}, there exists
a trivialization of a fibre bundle $V^*Q\to \mathbb R$ (not
necessarily compatible with its fibration $V^*Q\to Q$) such that
\be
\g_H=\dr_t, \qquad H=\wt p_id\wt q^i
\ee
with respect to the associated bundle coordinates $(t,\wt q^i, \wt
p_i)$. A direct computation shows that the Hamilton vector field
$\g_H$ (\ref{z3}) is an infinitesimal generator of a one-parameter
group of automorphisms of a Poisson manifold $(V^*Q,\{,\}_V)$.
Then one can show that $(t,\wt q^i,\wt p_i)$ are canonical
coordinates for the Poisson bracket $\{,\}_V$ \cite{book98}. Since
$\cH=0$, the Hamilton equation (\ref{z20a}) in these coordinates
takes a form
\be
\wt q^i_t=0, \qquad \wt p_{ti}=0,
\ee
i.e., $(t,\wt q^i,\wt p_i)$ are the initial data coordinates.

\begin{remark} \label{nn182}
In applications, the condition of the Hamilton connection $\g_H$
(\ref{z3}) to be complete need not holds on the entire phase space
(Section 9). In this case, one consider its subsets, and sometimes
we have different Hamiltonian systems on different subsets of
$V^*Q$.
\end{remark}

As was mentioned above,  one can associate to any non-autonomous
Hamiltonian system on a phase space $V^*Q$ an equivalent
autonomous symplectic Hamiltonian system on the cotangent bundle
$T^*Q$ as follows (Theorem \ref{09121}).

Given a Hamiltonian system $(V^*Q,H)$, its Hamiltonian $\cH$
(\ref{ws513}) defines a function
\mar{mm16}\beq
\cH^*=\dr_t\rfloor(\Xi_T-\zeta^* (-h)^*\Xi_T))=p_0+h=p_0+\cH
\label{mm16}
\eeq
on $T^*Q$. Let us regard $\cH^*$ (\ref{mm16}) as a Hamiltonian of
an autonomous Hamiltonian system on a symplectic manifold
$(T^*Q,\Om_T)$. The corresponding autonomous Hamilton equation on
$T^*Q$ takes a form
\mar{z20'}\beq
\dot t=1, \qquad \dot p_0=-\dr_t\cH, \qquad \dot q^i=\dr^i\cH,
\qquad \dot p_i=-\dr_i\cH. \label{z20'}
\eeq

\begin{remark} \label{0170}
Let us note that the splitting $\cH^*=p_0+\cH$ (\ref{mm16}) is ill
defined. At the same time, any reference frame $\G$ yields the
decomposition
\be
\cH^*=(p_0+\cH_\G) + (\cH-\cH_\G) = \cH^*_\G +\cE_\G,
\ee
where $\cH_\G$ is the Hamiltonian (\ref{ws515}) and $\cE_\G$
(\ref{xx60}) is the Hamiltonian function relative to a reference
frame $\G$.
\end{remark}

A Hamiltonian vector field $\vt_{\cH^*}$ of the function $\cH^*$
(\ref{mm16}) on $T^*Q$ is
\be
\vt_{\cH^*}=\dr_t -\dr_t\cH\dr^0+ \dr^i\cH\dr_i- \dr_i\cH\dr^i,
\qquad \vt_{\cH^*}\rfloor\Om_T=-d\cH^*.
\ee
Written relative to the coordinates (\ref{09150}), this vector
field reads
\beq
\vt_{\cH^*}=\dr_t + \dr^i\cH\dr_i- \dr_i\cH\dr^i. \label{z5'}
\eeq
It is identically projected onto the Hamiltonian connection $\g_H$
(\ref{z3}) on $V^*Q$ such that
\beq
\zeta^*(\bL_{\g_H}f)=\{\cH^*,\zeta^*f\}_T, \qquad f\in
C^\infty(V^*Q). \label{ws525}
\eeq
Therefore, the Hamilton equation (\ref{z20a}) is equivalent to the
autonomous Hamilton equation (\ref{z20'}).

Obviously, the Hamiltonian vector field $\vt_{\cH^*}$ (\ref{z5'})
is complete if the Hamilton vector field $\g_H$ (\ref{z3}) is so.

Thus, the following has been proved \cite{dew,book10,mang00}.

\begin{theorem} \label{09121} \mar{09121}  A non-autonomous Hamiltonian system $(V^*Q,H)$
of $k$ degrees of freedom is equivalent to an autonomous
Hamiltonian system $(T^*Q,\cH^*)$ of $k+1$ degrees of freedom on a
symplectic manifold $(T^*Q,\Om_T)$ whose Hamiltonian is the
function $\cH^*$ (\ref{mm16}).
\end{theorem}

We agree to call $(T^*Q,\cH^*)$ the homogeneous Hamiltonian system
and $\cH^*$ (\ref{mm16}) the homogeneous Hamiltonian.

It is readily observed that the Hamiltonian form $H$ (\ref{b4210})
also is the Poincar\'e--Cartan form (\ref{303'}) of the
characteristic Lagrangian
\mar{Q33}\beq
L_H=h_0(H) = (p_iq^i_t - \cH)dt \label{Q33}
\eeq
on the jet manifold $J^1V^*Q$ of $V^*Q\to\mathbb R$.

\begin{remark} \label{0110}
In fact, the Lagrangian (\ref{Q33}) is the pull-back onto
$J^1V^*Q$ of an exterior form $L_H$ on a product $V^*Q\times_Q
J^1Q$.
\end{remark}

The Lagrange operator (\ref{305'}) associated to the
characteristic Lagrangian $L_H$ (\ref{Q33}) reads
\be
\cE_H=\dl L_H=[(q^i_t-\dr^i\cH) dp_i -(p_{ti}+\dr_i\cH) dq^i]\w
dt.
\ee
The corresponding Lagrange equation (\ref{b327'}) is of first
order, and it coincides with the Hamilton equation (\ref{z20a}) on
$V^*Q$.

Due to this fact, Hamiltonian mechanics can be formulated as a
specific Lagrangian mechanics on a configuration space $V^*Q$.

In particular, let
\mar{z372}\beq
u=u^t\dr_t + u^i(t,q^j)\dr_i, \qquad u^t=0,1, \label{z372}
\eeq
be a vector field on a configuration space $Q$. Its canonical
functorial lift onto the cotangent bundle $T^*Q$ is
\beq
\wt u=u^t\dr_t + u^i\dr_i - p_j\dr_i u^j \dr^i. \label{gm513}
\eeq
This vector field is identically projected onto a vector field,
also given by the expression (\ref{gm513}), on a phase space
$V^*Q$ as a base of the trivial fibre bundle (\ref{nn169}). Then
we have the equality
\be
\bL_{\wt u}H= \bL_{J^1\wt u}L_H= (-u^t\dr_t\cH+p_i\dr_tu^i
-u^i\dr_i\cH +p_i\dr_j u^i\dr^j\cH)dt
\ee
for any Hamiltonian form $H$ (\ref{b4210}). This equality enables
us to study conservation laws in Hamiltonian mechanics similarly
to those in Lagrangian mechanics (Section 7).

Lagrangian and Hamiltonian formulations of mechanics as like as
those of field theory fail to be equivalent, unless Lagrangians
are hyperregular \cite{book09,sard15,book16}. The comprehensive
relations between Lagrangian and Hamiltonian systems can be
established in the case of almost regular Lagrangians
\cite{book10,mang00,sard98}. This is a particular case of the
relations between Lagrangian and covariant Hamiltonian theories on
fibre bundles \cite{book09,sard15}.

If the Lagrangian $L$ (\ref{23f2'}) (Definition \ref{d11}) is
hyperregular, it admits a unique associated Hamiltonian form
\beq
H=p_idq^i -(p_i\wh L^{-1i} - \cL(t, q^j,\wh L^{-1j}))dt.
\label{cc311b}
\eeq
Let $s$ be a classical solution of the Lagrange equation
(\ref{b327'}) for a Lagrangian $L$. A direct computation shows
that $\wh L\circ J^1s$ is a classical solution of the Hamilton
equation (\ref{z20a}) for the Hamiltonian form $H$ (\ref{cc311b}).
Conversely, if $r$ is a classical solution of the Hamilton
equation (\ref{z20a}) for the Hamiltonian form $H$ (\ref{cc311b}),
then $s=\pi_\Pi\circ r$ is a solution of the Lagrange equation
(\ref{b327'}) for $L$.

Let us restrict our consideration to almost regular Lagrangians
$L$ (Definition \ref{d11}).

\begin{theorem}\label{d3.23}
Let a section $r$ of $V^*Q\to \mathbb R$ be a  classical solution
of the Hamilton equation (\ref{z20a}) for a Hamiltonian form $H$
weakly associated to an almost regular Lagrangian $L$
\cite{book10,book16}. If $r$ lives in the Lagrangian constraint
space $N_L$, a section $s=\pi\circ r$ of $\pi:Q\to \mathbb R$
satisfies the Lagrange equation (\ref{b327'}), while $\ol s=\wh
H\circ r$, where
\be
\wh H: V^*Q\ar_Q J^1Q, \qquad q^i_t\circ\wh H=\dr^i\cH
\ee
is a Hamiltonian map, obeys the Cartan equation (\ref{b336c}).
\end{theorem}

\begin{theorem}\label{d3.24}  Given an almost regular
Lagrangian $L$, let a section $\ol s$ of the jet bundle $J^1Q\to
\mathbb R$ be a solution of the Cartan equation  (\ref{b336c}).
Let $H$ be a Hamiltonian form weakly associated to $L$, and let
$H$ satisfy a relation
\beq
\wh H\circ \wh L\circ \ol s=J^1s, \label{2.36q}
\eeq
where $s$ is the projection of $\ol s$ onto $Q$. Then a section
$r=\wh L\circ \ol s$ of a fibre bundle $V^*Q\to \mathbb R$ is a
classical solution of the Hamilton equation (\ref{z20a}) for $H$.
\end{theorem}

A set of Hamiltonian forms $H$ weakly associated to an almost
regular Lagrangian $L$ is said to be complete  if, for each
classical solution $s$ of a Lagrange equation, there exists a
classical solution $r$ of a Hamilton equation for a Hamiltonian
form $H$ from this set such that $s=\pi_\Pi\circ r$. By virtue of
Theorem \ref{d3.24}, a set of weakly associated Hamiltonian forms
is complete if, for every classical solution $s$ of a Lagrange
equation for $L$, there exists a Hamiltonian form $H$ from this
set which fulfills the relation (\ref{2.36q}) where $\ol s=J^1s$,
i.e.,
\beq
\wh H\circ \wh L\circ J^1s=J^1s. \label{072}
\eeq

\section{Hamiltonian conservation laws: Noether's inverse first theorem}

As was mentioned above, integrals of motion in Lagrangian
mechanics can come from Lagrangian symmetries (Theorem
\ref{035'}), but not any integral of motion is of this type. In
Hamiltonian mechanics, all integrals of motion are conserved
currents (Theorem \ref{0150'}). One can think of this fact as
being Noether's inverse first theorem.

\begin{definition} \label{nn232} \mar{nn232} An integral of
motion of a Hamiltonian system $(V^*Q,H)$ is defined as a smooth
real function $\Phi$ on $V^*Q$ which is an integral of motion of
the Hamilton equation (\ref{z20a}) in accordance with Definition
\ref{026}, i.e., it satisfies the relation (\ref{027}).
\end{definition}

Since the Hamilton equation (\ref{z20a}) is the kernel of the
covariant differential $D_{\g_H}$, this relation $d_t\Phi\ap 0$ is
equivalent to the equality
\mar{ws516}\beq
\bL_{\g_H} \Phi=(\dr_t +\g_H^i\dr_i +\g_{Hi}\dr^i)\Phi =\dr_t\Phi
+\{\cH,\Phi\}_V=0, \label{ws516}
\eeq
i.e., the Lie derivative of $\Phi$ along the Hamilton connection
$\g_H$ (\ref{z3}) vanishes.

At the same time, it follows from Theorem \ref{082} that  a vector
field $\up$ on $V^*Q$ is a symmetry of the Hamilton equation
(\ref{z20a}) in accordance with Definition \ref{025} if and only
if $[\g_H,\up]=0$. Given the Hamiltonian vector field $\vt_\Phi$
(\ref{m73}) of $\Phi$ with respect to the Poisson bracket
(\ref{m72}), it is easily justified that
\mar{092}\beq
[\g_H,\vt_\Phi]=\vt_{\bL_{\g_H} \Phi}. \label{092}
\eeq
Thus, we can conclude the following.

\begin{theorem} \label{nn192} \mar{nn192}
The Hamiltonian vector field of an integral of motion is a
symmetry of the Hamilton equation (\ref{z20a}).
\end{theorem}

Given a Hamiltonian system $(V^*Q,H)$, let $(T^*Q,\cH^*)$ be an
equivalent homogeneous Hamiltonian system. It follows from the
equality (\ref{ws525}) that
\beq
\zeta^*(\bL_{\g_H}\Phi)=\{\cH^*,\zeta^*\Phi\}_T =\zeta^*(\dr_t\Phi
+\{\cH,\Phi\}_V) \label{077a}
\eeq
for any function $\Phi\in C^\infty(V^*Q)$. This formula is
equivalent to that (\ref{ws516}).

\begin{theorem} \label{075} \mar{075} A function $\Phi\in
C^\infty(V^*Q)$ is an integral of motion of a Hamiltonian system
$(V^*Q,H)$ if and only if its pull-back $\zeta^*\Phi$ onto $T^*Q$
is an integral of motion of a homogeneous Hamiltonian system
$(T^*Q,\cH^*)$.
\end{theorem}

\begin{proof} The result follows from the equality (\ref{077a}):
\mar{077}\beq
\{\cH^*,\zeta^*\Phi\}_T=\zeta^*(\bL_{\g_H}\Phi)=0. \label{077}
\eeq
\end{proof}

\begin{theorem} \label{076}  If $\F$ and $\F'$ are
integrals of motion of a Hamiltonian system, their Poisson bracket
$\{\F,\F'\}_V$ also is an integral of motion.
\end{theorem}

\begin{proof} This fact results from the equalities (\ref{094})
and (\ref{077}).
\end{proof}

Consequently, integrals of motion of a Hamiltonian system
$(V^*Q,H)$ constitute a real Lie subalgebra of a Poisson algebra
$C^\infty(V^*Q)$.

Let us turn to Hamiltonian conservation laws. We are based on the
fact that the Hamilton equation (\ref{z20a}) also is a Lagrange
equation of the characteristic Lagrangian $L_H$ (\ref{Q33}).
Therefore one can study conservation laws in Hamiltonian mechanics
on a phase space $V^*Q$ similarly to those in Lagrangian mechanics
on a configuration space $V^*Q$ \cite{book10,jmp07,book16}.

Since the Hamilton equation (\ref{z20a}) is of first order, we
restrict our consideration to classical symmetries, i.e., vector
fields on $V^*Q$.

\begin{definition} \label{nn234} \mar{nn234}
A vector field on a phase space $V^*Q$ of a Hamiltonian system
$(V^*Q,H)$ is said to be its Hamiltonian symmetry if it is a
Lagrangian symmetry of the characteristic Lagrangian $L_H$.
\end{definition}

Let
\beq
\up=\up^t\dr_t + \up^i\dr_i + \up_i\dr^i, \qquad \up^t=0,1,
\label{0121}
\eeq
be a vector field on a phase space $V^*Q$. Its prolongation onto
$V^*Q\times_Q J^1Q$ (Remark \ref{0110}) reads
\be
J^1\up= \up^t\dr_t + \up^i\dr_i + \up_i\dr^i + d_t\up^i\dr_i^t.
\ee
Then the first variational formula (\ref{J4'}) for the
characteristic Lagrangian $L_H$ (\ref{Q33}) takes a form
\mar{nn203}\ben
&& -\up^t\dr_t\cH - \up^i\dr_i\cH +\up_i(q^i_t -\dr^i\cH) +p_id_t\up^i
=\label{nn203}\\
&& \quad (q^i_t\up^t-\up^i)(p_{ti}+\dr_i\cH)+ (\up_i-p_{ti}\up^t)(q^i_t-\dr^i\cH)
 + d_t(p_i\up^i-\up^t\cH). \nonumber
\een

If $\up$ (\ref{0121}) is a symmetry of $L_H$, i.e.,
\be
\bL_{J^1\up} L_H=d_t\si dt,
\ee
we obtain the weak Hamiltonian conservation law (\ref{d22f2}):
\mar{nn190}\beq
0\ap -d_t \cJ \label{nn190}
\eeq
of the Hamiltonian symmetry current (\ref{am225'}):
\mar{nn200}\beq
\cJ_\up=-p_i\up^i+\up^t\cH+\si. \label{nn200}
\eeq

The vector field $\up$ (\ref{0121}) on $V^*Q$ is a symmetry of the
characteristic Lagrangian $L_H$ (\ref{Q33}) if and only if
\beq
\up^i(p_{ti}+\dr_i\cH) -\up_i(q^i_t-\dr^i\cH) +
\up^t\dr_t\cH=d_t(-\cJ_\up +\up^t\cH). \label{0140}
\eeq
A glance at this equality shows the following.

\begin{theorem} \label{0146} \mar{0146}
The vector field $\up$ (\ref{0121}) is a Hamiltonian symmetry in
accordance with Definition \ref{nn234} only if
\mar{nn208}\beq
\dr^i\up_i=-\dr_i\up^i. \label{nn208}
\eeq
\end{theorem}

\begin{remark} \label{nn201} \mar{nn201} It is readily observed
that the Hamiltonian connection $\g_H$ (\ref{z3}) is a symmetry of
the characteristic Lagrangian $L_H$ whose conserved Hamiltonian
current (\ref{nn200}) equals zero. It follows that, given a
non-vertical Hamiltonian symmetry $\up$, $\up^t=1$, there exists a
vertical Hamiltonian symmetry $\up-\g_H$ with the same conserved
Hamiltonian current as $\up$.
\end{remark}

By virtue of Theorem \ref{nn191}, any Hamiltonian symmetry, being
classical symmetry of the characteristic Lagrangian $L_H$
(\ref{Q33}), also is symmetry of the Hamilton equation
(\ref{z20a}). In accordance with Theorem \ref{035'}, the
corresponding conserved Hamiltonian current (\ref{nn200}) is an
integral of motion of a Hamiltonian system which, thus, comes from
its Hamiltonian symmetry.

The converse also is true.

\begin{theorem} \label{0150'} \mar{0150'}
Any integral of motion $\Phi$ of a Hamiltonian system $(V^*Q,H)$
is the conserved Hamiltonian current $\cJ_{-\vt_\Phi}$
(\ref{nn200}) along the Hamiltonian vector field $-\vt_\Phi$
(\ref{m73}) of $-\Phi$.
\end{theorem}

\begin{proof}
It follows from the relations (\ref{ws516}) and (\ref{nn203}) that
\be
\bL_{-J^1\vt_\Phi}=d_t(\Phi- p_i\dr^i\Phi).
\ee
Then the equality (\ref{nn200}) results in a desired relation
$\Phi=\cJ_{-\vt_\Phi}$.
\end{proof}

This assertion can be regarded as above mentioned Noether's
inverse first theorem.

For instance, if the Hamiltonian symmetry $\up$ (\ref{0121}) is
projectable onto $Q$ (i.e., its components $\up^i=u^i$ are
independent of momenta $p_i$), then we $\up_i=-p_j\dr_i u^j$ in
accordance with the equality (\ref{nn208}). Consequently, $\up$ is
the canonical lift $\wt u$ (\ref{gm513}) onto $V^*Q$ of the vector
field $u$ (\ref{z372}) on $Q$. If $\wt u$ is a symmetry of the
characteristic Lagrangian $L_H$, it follows at once from the
equality (\ref{0140}) that $\wt u$ is an exact symmetry of $L_H$.
The corresponding conserved Hamiltonian symmetry current
 (\ref{nn200}) reads
\mar{0150}\beq
\wt\cJ_u=\cJ_{\wt u}=-p_iu^i+u^t\cH. \label{0150}
\eeq

\begin{definition} \label{nn235} \mar{nn235}
The vector field $u$ (\ref{z372}) on a configuration space $Q$ is
said to be the basic Hamiltonian symmetry if its canonical lift
$\wt u$ (\ref{gm513}) onto $V^*Q$ is a Hamiltonian symmetry.
\end{definition}

If a basic Hamiltonian symmetry $u$ is vertical, the corresponding
conserved Hamiltonian symmetry current (\ref{0150}):
\mar{nn206}\beq
\wt\cJ_u=-p_iu^i, \label{nn206}
\eeq
is a Noether current.

Now let $\G$ be the connection (\ref{a1.10}) on $Q$. The
corresponding symmetry current (\ref{0150}) is the Hamiltonian
function (\ref{xx60}):
\mar{nn205}\beq
\wt\cJ_\G=\cJ_{\wt\G}=\cE_\G =\cH-p_i\G^i, \label{nn205}
\eeq
relative to a reference frame $\G$. Given bundle coordinates
adapted to $\G$, we obtain the Lie derivative
\be
\bL_{J^1\wt\G}L_H=-\dr_t\cH.
\ee
It follows that a connection $\G$ is a basic Hamiltonian symmetry
if and only if the Hamiltonian $\cH$ (\ref{ws513}), written with
respect to the coordinates adapted to $\G$, is time-independent.
In this case, the Hamiltonian function (\ref{nn205}) is an
integral of motion of a Hamiltonian system.

There is the following relation between Lagrangian symmetries and
basic Hamiltonian symmetries if they are the same vector fields on
a configuration space $Q$.

\begin{theorem}\label{hamlaw'}  Let a
Hamiltonian form $H$ be associated with an almost regular
Lagrangian $L$. Let $r$ be a solution of the Hamilton equation
(\ref{z20a}) for $H$ which lives in the Lagrangian constraint
space $N_L$ (\ref{jkl}). Let $s=\pi_\Pi\circ r$ be the
corresponding solution of a Lagrange equation for $L$ so that the
relation (\ref{072}) holds. Then, for any vector field $u$
(\ref{z372}) on a fibre bundle $Q\to \mathbb R$, we have
\mar{Q10a}\beq
\wt\cJ_u (r)=\cJ_u( \pi_\Pi\circ r),\qquad \wt\cJ_u (\wh L\circ
J^1s) =\cJ_u(s), \label{Q10a}
\eeq
where $\cJ_u$ is the symmetry current (\ref{m225}) on $J^1Y$  and
$\wt\cJ_u=\cJ_{\wt u}$ is the symmetry current (\ref{0150}) on
$V^*Q$.
\end{theorem}

By virtue of Theorems \ref{d3.23} -- \ref{d3.24}, it follows that:

$\bullet$ if $\cJ_u$ in Theorem \ref{hamlaw'} is a conserved
symmetry current, then the symmetry current $\wt\cJ_u$
(\ref{Q10a}) is conserved on solutions of a Hamilton equation
which live in the Lagrangian constraint space;

$\bullet$ if $\wt\cJ_u$ in Theorem \ref{hamlaw'} is a conserved
symmetry current, then the symmetry current $\cJ_u$ (\ref{Q10a})
is conserved on solutions $s$ of a Lagrange equation which obey
the condition (\ref{072}).

In particular, let $u=\G$ be a connection and $E_\G$ the energy
function (\ref{m228}). Then the relations (\ref{Q10a}):
\be
&& \cE_\G(r)=\wt\cJ_\G (r)=\cJ_\G( \pi_\Pi\circ r)=E_\G(\pi_\Pi\circ r),\\
&&\cE_\G(\wh L\circ J^1s)= \wt\cJ_\G (\wh L\circ J^1s)
=\cJ_\G(s)=E_\G(s),
\ee
show that the Hamiltonian function $\cE_\G$ (\ref{nn205}) can be
treated as a Hamiltonian energy function relative to a reference
frame $\G$.

\section{Completely integrable Hamiltonian systems}

It may happen that symmetries and the corresponding integrals of
motion define a Hamiltonian system in full. This is the case of
commutative and noncommutative completely integrable systems
(henceforth, CISs) (Definition \ref{nn209}).

In view of Remark \ref{nn201}, we can restrict our consideration
to vertical symmetries $\up$ (\ref{0121}) where $\up^t=0$.

\begin{definition} \label{nn209} \mar{nn209} A non-autonomous Hamiltonian
system $(V^*Q,H)$ of $n=\di Q-1$ degrees of freedom is said to be
completely integrable if it admits $n\leq k<2n$ vertical classical
symmetries $\up_\al$ which obey the following conditions.

(i) Symmetries $\up_\al$ everywhere are linearly independent.

(ii) They form a $k$-dimensional real Lie algebra $\gG$ of corank
$m=2n-k$ with commutation relations
\mar{nn210}\beq
[\up_\al,\up_\bt]=c^\nu_{\al\bt}\up_\nu. \label{nn210}
\eeq
\end{definition}

If $k=n$, then a Lie algebra $\gG$ is commutative, and we are in
the case of a commutative CIS. If $n< k$, the Lie algebra
(\ref{nn210}) is noncommutative, and a CIS is called
noncommutative or superintegrable.

The conditions of Definition \ref{nn209} can be reformulated in
terms of integrals of motion $\Phi_\al=-\cJ_{\up_\al}$
corresponding to
 symmetries $\up_\al$. By virtue of Noether's inverse first
Theorem \ref{0150'}, $\up_\al=\vt_{\Phi_\al}$ are the Hamiltonian
vector fields (\ref{m73}) of integrals of motion $\Phi_\al$. In
accordance with the relation (\ref{094}), integrals of motion obey
the commutation relations
\mar{nn211}\beq
\{\Phi_\al,\Phi_\bt\}_V=c_{\al\bt}^\nu \Phi_\nu. \label{nn211}
\eeq
Then we come to an equivalent definition of a CISs
\cite{book10,ijgmmp12,book15}.

\begin{definition} \label{i0x} \mar{i0x}
A non-autonomous Hamiltonian system $(V^*Q,H)$ of $n=\di Q-1$
degrees of freedom is a CIS if it possesses $n\leq k<2n$ integrals
of motion $\F_1,\ldots,\F_k$, obeying the following conditions.

(i) All the functions $\F_\al$ are independent, i.e., a $k$-form
$d\F_1\w\cdots\w d\F_k$ nowhere vanishes on $V^*Q$. It follows
that a map
\beq
\F:V^*Q\to N=(\F_1(V^*Q),\ldots,\F_k(V^*Q))\subset \mathbb R^k
\label{nc4x}
\eeq
is a fibred manifold over a connected open subset $N\subset\mathbb
R^k$.

(ii) The commutation relations (\ref{nn211}) are satisfied.
\end{definition}

Given a non-autonomous CIS in accordance with Definition
\ref{i0x}, the equivalent autonomous Hamiltonian system on a
homogeneous phase space $T^*Q$ (Theorem \ref{09121}) possesses
$k+1$ integrals of motion
\mar{09136}\beq
(\cH^*,\zeta^*\F_1,\ldots,\zeta^*\F_k) \label{09136}
\eeq
with the following properties (Theorem \ref{075}).

(i) The integrals of motion (\ref{09136}) are mutually
independent, and a map
\ben
&& \wt\F:T^*Q\to
(\cH^*(T^*Q),\zeta^*\F_1(T^*Q),\ldots,\zeta^*\F_k(T^*Q))= \label{09140}\\
&& \qquad (I_0,\F_1(V^*Q),\ldots,\F_k(V^*Q))= \mathbb R\times
N=N'\nonumber
\een
is a fibred manifold.

(ii) The integrals of motion (\ref{09136}) obey the commutation
relations
\be
\{\zeta^*\F_\al,\zeta^*\F_\bt\}= c_{\al\bt}^\nu
\zeta^*\F_\bt,\qquad \{\cH^*,\zeta^*\F_\al\}=0.
\ee
They generate a real $(k+1)$ dimensional Lie algebra of corank
$2n+1-k$.

As a result, integrals of motion (\ref{09136}) form an autonomous
CIS on a symplectic manifold $(T^*Q, \Om_T)$ in accordance with
Definition \ref{i0}. In order to describe it, one then can follow
the Mishchenko--Fomenko theorem \cite{bols03,fasso05,mishc}
extended to the case of noncompact invariant submanifolds
\cite{fior2,ijgmmp09a,book15}.

Therefore, let us turn to CISs (superintegrable systems) on a
symplectic manifold.

\begin{remark} \label{nn236} \mar{nn236}
Let $Z$ be a manifold. Any exterior two-form $\Om$ on $Z$ yields a
linear bundle morphism
\beq
\Om^\flat: TZ\op\to_Z T^*Z, \qquad \Om^\flat:v \to -
v\rfloor\Om(z), \quad v\in T_zZ, \quad z\in Z. \label{m52}
\eeq
One says that a two-form $\Om$ is non-degenerate if $\Ker
\Om^\flat=0$. A closed non-degenerate form is called the
symplectic form. Accordingly, a manifold $Z$ equipped with a
symplectic form is said to be the symplectic manifold. A
symplectic manifold necessarily is even-dimensional. A closed
two-form on $Z$ is called presymplectic  if it is not necessary
degenerate. A vector field $u$ on a symplectic manifold $(Z,\Om)$
is said to be Hamiltonian  if a one-form $u\rfloor\Om$ is exact.
Any smooth function $f\in C^\infty(Z)$ on $Z$ defines a unique
Hamiltonian vector field $\vt_f$, called the Hamiltonian vector
field of a function $f$, such that
\mar{z100}\beq
\vt_f\rfloor\Om=-df, \qquad \vt_f=\Om^\sharp(df),  \label{z100}
\eeq
where  $\Om^\sharp$ is the inverse isomorphism to $\Om^\flat$
(\ref{m52}). Given an $m$-dimensional manifold $M$ coordinated by
$(q^i)$, let $T^*M$ be its cotangent bundle equipped with the
holonomic coordinates $(q^i, \dot q_i)$. It is endowed with the
canonical Liouville form
\mar{nn237}\beq
\Xi_T=\dot q_idq^i \label{nn237}
\eeq
and the canonical symplectic form
\mar{m83}\beq
\Om_T= d\Xi=d\dot q_i\w dq^i. \label{m83}
\eeq
The Hamiltonian vector field $\vt_f$ (\ref{z100}) with respect to
the canonical symplectic form (\ref{m83}) reads
\mar{spr829}\beq
\vt_f=\dr^if\dr_i -\dr_i f\dr^i. \label{spr829}
\eeq
A symplectic form $\Om$ on a manifold $Z$ defines a Poisson
bracket
\be
\{f,g\}=\vt_g\rfloor\vt_f\rfloor\Om, \qquad f,g\in C^\infty(Z).
\ee
The canonical symplectic form $\Om_T$ (\ref{m83}) on $T^*M$ yields
the canonical Poisson bracket
\mar{nn238}\beq
\{f,g\}_T=\frac{\dr f}{\dr \dot q_i}\frac{\dr g}{\dr q^i}-
\frac{\dr f}{\dr q^i}\frac{\dr g}{\dr \dot q_i}. \label{nn238}
\eeq
\end{remark}

\begin{definition} \label{i0} \mar{i0}
Let $(Z,\Om)$ be a $2n$-dimensional connected symplectic manifold,
and let $(C^\infty(Z), \{,\})$ be a Poisson algebra of smooth real
functions on $Z$. A subset
\mar{i00}\beq
F=(F_1,\ldots,F_k), \qquad n\leq k<2n, \label{i00}
\eeq
of a Poisson algebra $C^\infty(Z)$ is called the CIS  or the
superintegrable system  if the following conditions hold.

(i) All the functions $F_i$ (called the generating functions of a
CIS)  are independent, i.e., a $k$-form $\op\w^kdF_i$ nowhere
vanishes on $Z$. It follows that a map $F:Z\to \mathbb R^k$ is a
submersion, i.e.,
\mar{nc4}\beq
F:Z\to N=F(Z) \label{nc4}
\eeq
is a fibred manifold over a domain $N\subset\mathbb R^k$ endowed
with the coordinates $(x_i)$ such that $x_i\circ F=F_i$.

(ii) There exist smooth real functions $s_{ij}$ on $N$ such that
\beq
\{F_i,F_j\}= s_{ij}\circ F, \qquad i,j=1,\ldots, k. \label{nc1}
\eeq

(iii) The $(k\times k)$-matrix function $\mathbf{s}$ with the
entries $s_{ij}$ (\ref{nc1}) is of constant corank $m=2n-k$ at all
points of $N$.
\end{definition}

\begin{remark} \label{cmp21a}
If $k=n$, then $\mathbf{s}=0$, and we are in the case of
commutative CISs when $F_1,...,F-n$ are independent functions in
involution.
\end{remark}

If $k>n$, the matrix $\mathbf{s}$ is necessarily nonzero. If
$k=2n-1$, a CIS is called maximally integrable.

The following two assertions clarify a structure of CISs
\cite{fasso05,fior2,book15}.

\begin{theorem} \label{nc7}  \mar{nc7} Given a symplectic manifold $(Z,\Om)$,
let $F:Z\to N$ be a fibred manifold such that, for any two
functions $f$, $f'$ constant on fibres of $F$, their Poisson
bracket $\{f,f'\}$ is so. By virtue of the well known theorem
\cite{book10,vais}, $N$ is provided with an unique coinduced
Poisson structure $\{,\}_N$ such that $F$ is a Poisson morphism.
\end{theorem}

Since any function constant on fibres of $F$ is the pull-back of
some function on $N$, the CIS (\ref{i00}) satisfies the condition
of Theorem \ref{nc7} due to item (ii) of Definition \ref{i0}.
Thus, a base $N$ of the fibration (\ref{nc4}) is endowed with a
coinduced Poisson structure of corank $m$. With respect to
coordinates $x_i$ in item (i) of Definition \ref{i0} its bivector
field reads
\mar{cmp1}\beq
w=s_{ij}(x_k)\dr^i\w\dr^j. \label{cmp1}
\eeq

\begin{theorem} \label{nc8} \mar{nc8}  Given a fibred manifold $F:Z\to N$ in
Theorem \ref{nc7}, the following conditions are equivalent
\cite{fasso05}:

(i) a rank of the coinduced Poisson structure $\{,\}_N$ on $N$
equals $2\di N-\di Z$,

(ii) the fibres of $F$ are isotropic,

(iii) the fibres of $F$ are  maximal integral manifolds of the
involutive distribution spanned by the Hamiltonian vector fields
of the pull-back $F^*C$ of Casimir functions $C$ of the coinduced
Poisson structure (\ref{cmp1}) on $N$.
\end{theorem}

It is readily observed that the fibred manifold $F$ (\ref{nc4})
obeys condition (i) of Theorem \ref{nc8} due to item (iii) of
Definition \ref{i0}, namely, $k-m= 2(k-n)$.

Fibres of the fibred manifold $F$ (\ref{nc4}) are called the
invariant submanifolds.

\begin{remark} \label{cmp8} \mar{cmp8}
In practice, condition (i) of Definition \ref{i0} fails to hold
everywhere. It can be replaced with that a subset $Z_R\subset Z$
of regular points (where $\op\w^kdF_i\neq 0$) is open and dense.
Let $M$ be an invariant submanifold through a regular point $z\in
Z_R\subset Z$. Then it is regular, i.e., $M\subset Z_R$. Let $M$
admit a regular open saturated neighborhood  $U_M$ (i.e., a fibre
of $F$ through a point of $U_M$ belongs to $U_M$). For instance,
any compact invariant submanifold $M$ has such a neighborhood
$U_M$. The restriction of functions $F_i$ to $U_M$ defines a CIS
on $U_M$ which obeys Definition \ref{i0}. In this case, one says
that a CIS is considered around its invariant submanifold $M$.
\end{remark}

Let $(Z,\Om)$ be a $2n$-dimensional connected symplectic manifold.
Given the CIS $(F_i)$ (\ref{i00}) on $(Z,\Om)$, the well known
Mishchenko--Fomenko theorem (Theorem \ref{nc0}) states the
existence of action-angle coordinates around its connected compact
invariant submanifold \cite{bols03,fasso05,mishc}. This theorem
has been extended to CISs with noncompact invariant submanifolds
(Theorem \ref{nc0'}) \cite{fior2,ijgmmp09a,book15}. These
submanifolds are diffeomorphic to a toroidal cylinder
\mar{g120'}\beq
\mathbb R^{m-r}\times T^r, \qquad m=2n-k, \qquad 0\leq r\leq m.
\label{g120'}
\eeq

\begin{theorem} \label{nc0'} \mar{nc0'} Let the Hamiltonian vector fields $\vt_i$ of the
functions $F_i$ be complete, and let the fibres of the fibred
manifold $F$ (\ref{nc4}) be connected and mutually diffeomorphic.
Then the following hold.

(I) The fibres of $F$ (\ref{nc4}) are diffeomorphic to the
toroidal cylinder (\ref{g120'}).

(II) Given a fibre $M$ of $F$ (\ref{nc4}), there exists its open
saturated neighborhood $U_M$ which is a trivial principal bundle
\beq
U_M=N_M\times \mathbb R^{m-r}\times T^r\ar^F N_M \label{kk600}
\eeq
with the structure group (\ref{g120'}).

(III) A neighborhood $U_M$ is provided with the bundle
action-angle coordinates  $(I_\la,p_s,q^s, y^\la)$, $\la=1,\ldots,
m$, $s=1,\ldots,n-m$, such that: (i) the angle coordinates
$(y^\la)$ are those on a toroidal cylinder, i.e., fibre
coordinates on the fibre bundle (\ref{kk600}), (ii)
$(I_\la,p_s,q^s)$ are coordinates on its base $N_M$ where the
action coordinates $(I_\la)$ are values of Casimir functions of
the coinduced Poisson structure $\{,\}_N$ on $N_M$, and (iii) a
symplectic form $\Om$ on $U_M$ reads
\mar{cmp06}\beq
\Om= dI_\la\w dy^\la + dp_s\w dq^s. \label{cmp6}
\eeq
\end{theorem}

\begin{remark} \label{zz90}  The condition of the completeness of Hamiltonian
vector fields of the generating functions $F_i$ in Theorem
\ref{nc0'} is rather restrictive. One can replace this condition
with that the Hamiltonian vector fields of the pull-back onto $Z$
of Casimir functions on $N$ are complete.
\end{remark}

If the conditions of Theorem \ref{nc0'} are replaced with that
fibres of the fibred manifold $F$ (\ref{nc4}) are compact and
connected, this theorem restarts the Mishchenko--Fomenko theorem
as follows.

\begin{theorem} \label{nc0} \mar{nc0}
Let the fibres of the fibred manifold $F$ (\ref{nc4}) be connected
and compact. Then they are diffeomorphic to a torus $T^m$, and
statements (II) -- (III) of Theorem \ref{nc0'} hold.
\end{theorem}

\begin{remark}
In Theorem \ref{nc0}, the Hamiltonian vector fields $\up_\la$ are
complete because fibres of the fibred manifold $F$ (\ref{nc4}) are
compact. As well known, any vector field on a compact manifold is
complete.
\end{remark}

To study a CIS, one conventionally considers it with respect to
action-angle coordinates. A problem is that an action-angle
coordinate chart on an open subbundle $U$ of the fibred manifold
$Z\to N$ (\ref{nc4}) in Theorem \ref{nc0'} is local. The following
generalizes this theorem to the case of global action-angle
coordinates.

\begin{definition}  \label{cmp30a} \mar{cmp30a} The CIS
$F$ (\ref{i00}) on a symplectic manifold $(Z,\Om)$ in Definition
\ref{i0} is called globally integrable (or, shortly, global) if
there exist global action-angle coordinates
\beq
(I_\la, x^A, y^\la), \qquad \la=1,\ldots,m, \qquad A=1,\ldots,
2(n-m), \label{cmp31a}
\eeq
such that: (i) the action coordinates $(I_\la)$ are expressed in
values of some Casimir functions $C_\la$ on a Poisson manifold
$(N,\{,\}_N)$, (ii) the angle coordinates $(y^\la)$ are
coordinates on the toroidal cylinder $\mathbb R^{m-r}\times T^r$,
$0\leq r\leq m$, and (iii) a symplectic form $\Om$ on $Z$ reads
\beq
\Om= dI_\la\w d y^\la +\Om_{AB}(I_\m,x^C) dx^A\w dx^B.
\label{cmp32}
\eeq
\end{definition}

It is readily observed that the action-angle coordinates on $U$ in
Theorem \ref{nc0'} are global on $U$ in accordance with Definition
\ref{cmp30a}.

Forthcoming Theorem \ref{cmp34} provides the sufficient conditions
of the existence of global action-angle coordinates of a CIS on a
symplectic manifold $(Z,\Om)$
\cite{book10,jmp07,ijgmmp09a,book15}. It generalizes the
well-known result for the case of compact invariant submanifolds
\cite{daz,fasso05}.

\begin{theorem} \label{cmp34}  \mar{cmp34} A CIS $F$ on
a symplectic manifold $(Z,\Om)$ is globally integrable if the
following conditions hold.

(i) Hamiltonian vector fields $\vt_i$ of the generating functions
$F_i$ are complete.

(ii) The fibred manifold $F$ (\ref{nc4}) is a fibre bundle with
connected fibres.

(iii) Its base $N$ is simply connected and the cohomology
$H^2(N;\mathbb Z)$ with coefficients in the constant sheaf
$\mathbb Z$ is trivial.

(iv) The coinduced Poisson structure $\{,\}_N$ on a base $N$
admits $m$ independent Casimir functions $C_\la$.
\end{theorem}

Theorem \ref{cmp34} restarts Theorem \ref{nc0'} if one considers
an open subset $V$ of $N$ admitting the Darboux coordinates $x^A$
on symplectic leaves of $U$. If invariant submanifolds of a CIS
are assumed to be compact, condition (i) of Theorem \ref{cmp34} is
unnecessary since vector fields $\vt_\la$ on compact fibres of $F$
are complete. In this case, Theorem \ref{cmp34} reproduces the
well known result in \cite{daz}.

Furthermore, one can show that condition (iii) of Theorem
\ref{cmp34} guarantee that fibre bundles $F$ in conditions (ii) of
these theorems are trivial \cite{book15}. Therefore, Theorem
\ref{cmp34} can be reformulated as follows.

\begin{theorem}  \label{cmp36} \mar{cmp36} A CIS $F$ on
a symplectic manifold $(Z,\Om)$ is global if and only if the
following conditions hold.

(i) The fibred manifold $F$ (\ref{nc4}) is a trivial fibre bundle.

(ii) The coinduced Poisson structure $\{,\}_N$ on a base $N$
admits $m$ independent Casimir functions $C_\la$ such that
Hamiltonian vector fields of their pull-back $F^*C_\la$ are
complete.
\end{theorem}

Bearing in mind the autonomous CIS (\ref{09136}), let us turn to
autonomous CISs whose generating functions are integrals of
motion, i.e., they are in involution with a Hamiltonian $\cH$, and
the functions $(\cH,F_1,\ldots,F_k)$ are nowhere independent,
i.e.,
\ben
&&\{\cH, F_i\}=0, \label{cmp11}\\
&& d\cH\w(\op\w^kdF_i)=0. \label{cmp11'}
\een
Let us note that, in accordance with item (ii) of Theorem
\ref{cmp36} and forthcoming Theorem \ref{cmp12b},  the Hamiltonian
vector field of a Hamiltonian $\cH$ of a CIS always is complete.

\begin{theorem} \label{cmp12b}
It follows from the equality (\ref{cmp11'}) that a Hamiltonian
$\cH$ is constant on the invariant submanifolds. Therefore, it is
the pull-back of a function on $N$ which is a Casimir function of
the Poisson structure (\ref{cmp1}) because of the conditions
(\ref{cmp11}).
\end{theorem}

Theorem \ref{cmp12b} leads to the following.

\begin{theorem} \label{zz1}
Let $\cH$ be a Hamiltonian of an autonomous global CIS provided
with the action-angle coordinates $(I_\la, x^A, y^\la)$
(\ref{cmp31a}). Then a Hamiltonian $\cH$ depends only on the
action coordinates $I_\la$. Consequently, the Hamilton equation of
a global CIS takes a form
\be
\dot y^\la=\frac{\dr \cH}{\dr I_\la}, \qquad
I_\la=\mathrm{const.}, \qquad x^A=\mathrm{const.}
\ee
\end{theorem}

\begin{remark}
Given a Hamiltonian $\cH$ of a Hamiltonian system on a symplectic
manifold $Z$, it may happen that we have different CISs on
different open subsets of $Z$. For instance, this is the case of
the global Kepler problem (Section 9).
\end{remark}

\begin{remark} \label{nn240} \mar{nn240}
Bearing in mind again the autonomous CIS (\ref{09136}), let us
also consider CISs whose generating functions $\{F_1,\ldots,F_k\}$
form a $k$-dimensional real Lie algebra $\cG$ of corank $m$ with
commutation relations
\beq
\{F_i,F_j\}= c_{ij}^h F_h, \qquad c_{ij}^h=\mathrm{const.}
\label{zz60}
\eeq
Then $F$ (\ref{nc4}) is a momentum mapping of $Z$ to the Lie
coalgebra $\cG^*$ provided with the coordinates $x_i$ in item (i)
of Definition \ref{i0} \cite{book05,guil}. In this case, the
coinduced Poisson structure $\{,\}_N$ coincides with the canonical
Lie--Poisson structure on $\cG^*$ given by the Poisson bivector
field
\be
w=\frac12 c_{ij}^h x_h\dr^i\w\dr^j.
\ee
Let $V$ be an open subset of $\cG^*$ such that conditions (i) and
(ii) of Theorem \ref{cmp36} are satisfied. Then an open subset
$F^{-1}(V)\subset Z$ is provided with the action-angle
coordinates. Let Hamiltonian vector fields $\vt_i$ of the
generating functions $F_i$ which form a Lie algebra $\cG$ be
complete. Then they define a locally free Hamiltonian action on
$Z$ of some simply connected Lie group $G$ whose Lie algebra is
isomorphic to $\cG$ \cite{palais}. Orbits of $G$ coincide with
$k$-dimensional maximal integral manifolds of the regular
distribution $\cV$ on $Z$ spanned by Hamiltonian vector fields
$\vt_i$ \cite{susm}. Furthermore, Casimir functions of the
Lie--Poisson structure on $\cG^*$ are exactly the coadjoint
invariant functions on $\cG^*$. They are constant on orbits of the
coadjoint action of $G$ on $\cG^*$ which coincide with leaves of
the symplectic foliation of $\cG^*$.
\end{remark}

Now, let us return to the autonomous CIS (\ref{09136}) on
homogeneous phase space of non-autonomous mechanics.

There is the commutative diagram
\be
\begin{array}{rcccl}
&T^*Q &\ar^\zeta & V^*Q&\\
_{\wt \F}& \put(0,10){\vector(0,-1){20}} & & \put(0,10){\vector(0,-1){20}}&_\F\\
& N' &\ar^\xi & N&
\end{array}
\ee
where $\zeta$ (\ref{b418a}) and $\xi:N'=\mathbb R\times N\to N$
are trivial bundles. It follows that the fibred manifold
(\ref{09140}) is the pull-back $\wt \F=\xi^* \F$ of the fibred
manifold $\F$ (\ref{nc4x}) onto $N'$.

Let the conditions of Theorem \ref{nc0'} hold. If the Hamiltonian
vector fields
\be
(\g_H,\vt_{\F_1},\ldots,\vt_{\F_k}),\qquad \vt_{\F_\al}=
\dr^i\F_\al\dr_i- \dr_i\F_\al\dr^i,
\ee
of integrals of motion $\F_\al$ on $V^*Q$ are complete, the
Hamiltonian vector fields
\be
(u_{\cH^*},u_{\zeta^*\F_1},\ldots,u_{\zeta^*\F_k}), \qquad
u_{\zeta^*\F_\al} = \dr^i\F_\al\dr_i- \dr_i\F_\al\dr^i,
\ee
on $T^*Q$ are complete. If fibres of the fibred manifold $\F$
(\ref{nc4x}) are connected and mutually diffeomorphic, the fibres
of the fibred manifold $\wt\F$ (\ref{09140}) also are well.

Let $M$ be a fibre of $\F$ (\ref{nc4x}) and $h(M)$ the
corresponding fibre of $\wt\F$ (\ref{09140}). In accordance with
Theorem \ref{nc0'}, there exists an open neighborhood $U'$ of
$h(M)$ which is a trivial principal bundle with the structure
group
\beq
\mathbb R^{1+m-r}\times T^r \label{g120z}
\eeq
whose bundle coordinates are the action-angle coordinates
\beq
(I_0,I_\la,t,y^\la,p_A,q^A), \qquad A=1,\ldots,n-m, \qquad
\la=1,\ldots, k,\label{09135'}
\eeq
such that:

(i) $(t,y^\la)$ are coordinates on the toroidal cylinder
(\ref{g120z}),

(ii) the symplectic form $\Om_T$ on $U'$ reads
\be
\Om_T= dI_0\w dt + dI_\al\w dy^\al + dp_A\w dq^A,
\ee

(iii) the action coordinates $(I_0,I_\al)$ are expressed in values
of the Casimir functions $C_0=I_0$, $C_\al$ of the coinduced
Poisson structure $w=\dr^A\w\dr_A$ on $N'$,

(iv) a homogeneous Hamiltonian $\cH^*$ depends on action
coordinates, namely, $\cH^*=I_0$,

(iv) the integrals of motion $\zeta^*\F_1, \ldots \zeta^*\F_k$ are
independent of coordinates $(t,y^\la)$.

Provided with the action-angle coordinates (\ref{09135'}), the
above mentioned neighborhood $U'$ is a trivial bundle $U'=\mathbb
R\times U_M$ where $U_M=\zeta(U')$ is an open neighborhood of the
fibre $M$ of the fibre bundle $\F$ (\ref{nc4x}). As a result, we
come to the following.

\begin{theorem} \label{nc0'x}
Let symmetries $\up_\al$ in Definition \ref{nn209} be complete,
and let fibres of the fibred manifold $\F$ (\ref{nc4x}) defined by
the corresponding conserved integrals of motion be connected and
mutually diffeomorphic. Then there exists an open neighborhood
$U_M$ of a fibre $M$ of $\F$ (\ref{nc4x}) which is a trivial
principal bundle with a structure group (\ref{g120z}) whose bundle
coordinates are the action-angle coordinates
\beq
(p_A,q^A,I_\la,t,y^\la), \qquad A=1,\ldots,k-n, \qquad
\la=1,\ldots, m,\label{09135}
\eeq
such that:

(i) $(t,y^\la)$ are coordinates on the toroidal cylinder
(\ref{g120z}),

(ii) the Poisson bracket $\{,\}_V$ on $U_M$ reads
\be
\{f,g\}_V = \dr^Af\dr_Ag-\dr^Ag\dr_Af + \dr^\la f\dr_\la g-\dr^\la
g\dr_\la f,
\ee

(iii) a Hamiltonian $\cH$ depends only on action coordinates
$I_\la$,

(iv) the integrals of motion $\F_1, \ldots \F_k$ are independent
of coordinates $(t,y^\la)$.
\end{theorem}

\section{Global Kepler problem}

We provides a global analysis of the Kepler problem as an example
of a mechanical system which is characterized by its symmetries in
full. It falls into two distinct global CISs on different open
subsets of a phase space. Their integrals of motion form the Lie
algebras $so(3)$ and $so(2,1)$ with compact and noncompact
invariant submanifolds, respectively
\cite{book10,ijgmmp09a,book15}.

Let us consider a mechanical system of a point mass in the
presence of a central potential. Its configuration space is
\beq
Q=\mathbb R\times\mathbb R^3\to\mathbb R \label{0155}
\eeq
endowed with the Cartesian coordinates $(t,q^i)$, $i=1,2,3$.

A Lagrangian of this mechanical system reads
\beq
\cL=\frac12 \left(\op\sum_i(q^i_t)^2\right) - V(r), \qquad
r=\left(\op\sum_i(q^i)^2\right)^{1/2}. \label{042}
\eeq

The vertical vector fields
\beq
v^a_b=q^b\dr_a -q^a\dr_b \label{043}
\eeq
on $Q$ (\ref{0155}) are infinitesimal generators of a group
$SO(3)$ acting on $\mathbb R^3$. Their jet prolongations
(\ref{a23f41}) read
\be
J^1v^a_b=q^b\dr_a -q^a\dr_b + q^b_t\dr_a^t -q^a_t\dr_b^t.
\ee
It is easily justified that the vector fields (\ref{043}) are
exact symmetries of the Lagrangian (\ref{042}). In accordance with
Noether's first theorem, the corresponding conserved Noether
currents (\ref{z384}) are orbital momenta
\beq
M^a_b=\cJ_{v^a_b} =(q^a \pi_b - q^b\pi_a)=q^aq^b_t- q^b q^a_t.
\label{045}
\eeq
They are integrals of motion, which however fail to be
independent.

Let us consider the Lagrangian system (\ref{042}) where
\beq
V(r)=-\frac1r \label{049}
\eeq
is the Kepler potential.  This Lagrangian system possesses
additional integrals of motion
\beq
A^a=\op\sum_b(q^aq^b_t -q^bq^a_t)q^b_t -\frac{q^a}{r}, \label{050}
\eeq
besides the orbital momenta (\ref{045}). They are components of
the Rung--Lenz vector.

However, there is no Lagrangian symmetry on $Q$ (\ref{0155}) whose
symmetry currents are $A^a$ (\ref{050}).

 Let us consider a Hamiltonian Kepler system on
the configuration space $Q$ (\ref{0155}). Its phase space is
$V^*Q=\mathbb R\times\mathbb R^6$ coordinated by $(t,q^i,p_i)$.

It is readily observed that the Lagrangian (\ref{042}) with the
Kepler potential (\ref{049}) of a Kepler system is hyperregular.
The associated Hamiltonian form reads
\beq
H=p_idq^i-\left[\frac12 \left(\op\sum_i(p_i)^2\right)
-\frac1r\right]dt. \label{0157}
\eeq
The corresponding characteristic Lagrangian $L_H$ (\ref{Q33}) is
\beq
L_H=\left[p_iq^i_t - \frac12 \left(\op\sum_i(p_i)^2\right)
+\frac1r\right]dt. \label{158}
\eeq

Then a Hamiltonian Kepler system possesses the following integrals
of motion:

$\bullet$ an energy function $\cE=\cH$;

$\bullet$ orbital momenta
\beq
M^a_b =q^a p_b - q^bp_a; \label{0159}
\eeq

$\bullet$ components of the Rung--Lenz vector
\beq
A^a=\op\sum_b(q^ap_b -q^bp_a)p_b -\frac{q^a}{r}. \label{0160'}
\eeq
By virtue of the Noether's inverse first Theorem \ref{0150'},
these integrals of motion are the conserved currents of the
following Hamiltonian symmetries:

$\bullet$ the exact symmetry $\dr_t$,

$\bullet$ the exact vertical symmetries
\beq
\up^a_b=q^b \dr_a - q^a\dr_b - p_a\dr^b + p_b\dr^a, \label{0161}
\eeq

$\bullet$ the vertical symmetries
\beq
\up^a=\op\sum_b[p_b\up^a_b + (q^bp_a -q^ap_b)\dr_b]
-\dr_b\left(\frac{q^a}{r}\right)\dr^b. \label{0162}
\eeq

Note that the Hamiltonian symmetries $\up^a_b$ (\ref{0161}) are
the canonical lift (\ref{gm513}) onto $V^*Q$ of the vector fields
$v^a_b$ (\ref{043}) on $Q$, which thus are basic Hamiltonian
symmetries, and integrals of motion $M^a_b$ (\ref{0159}) are the
Noether currents (\ref{nn206}).

At the same time, the Hamiltonian symmetries (\ref{0162}) do not
come from any vector fields on a configuration space $Q$.
Therefore, in contrast with the Rung--Lenz vector (\ref{0162}) in
Hamiltonian mechanics, the Rung--Lenz vector (\ref{050}) in
Lagrangian mechanics fails to be a conserved current of a
Lagrangian symmetry.

As was mentioned above, the Hamiltonian symmetries of the Kepler
problem make up CISs. To analyze them, we further consider the
Kepler problem on a configuration space $\mathbb R^2$ without a
loss of generality.

Its phase space is $T^*\mathbb R^2=\mathbb R^4$ provided with the
Cartesian coordinates $(q_i,p_i)$, $i=1,2$, and the canonical
symplectic form
\beq
\Om_T=\op\sum_idp_i\w dq_i. \label{zz43}
\eeq
Let us denote
\be
p=\left(\op\sum_i(p_i)^2\right)^{1/2}, \qquad
r=\left(\op\sum_i(q^i)^2\right)^{1/2}, \qquad
(p,q)=\op\sum_ip_iq_i.
\ee
An autonomous Hamiltonian of the Kepler system reads
\beq
\cH=\frac12p^2-\frac1r \label{zz41}
\eeq
(cf. (\ref{0157})). The Kepler system is a Hamiltonian system on a
symplectic manifold
\beq
Z=\mathbb R^4\setminus \{0\} \label{zz42}
\eeq
endowed with the symplectic form $\Om_T$ (\ref{zz43}).

Let us consider functions
\ben
&& M_{12}=-M_{21}=q_1p_2-q_2p_1,  \label{zz44}\\
&& A_i=\op\sum_j M_{ij}p_j -\frac{q_i}{r}=q_ip^2 -p_i(p,q)-\frac{q_i}{r}, \qquad i=1,2,
\label{zz45}
\een
on the symplectic manifold $Z$ (\ref{zz42}). As was mentioned
above, they are integrals of motion of the Hamiltonian $\cH$
(\ref{zz41}) where $M_{12}$ is an angular momentum and $(A_i)$ is
the Rung--Lenz vector. Let us denote
\beq
M^2=(M_{12})^2, \qquad  A^2=(A_1)^2 + (A_a)^2=2M^2\cH+1.
\label{zz50}
\eeq

Let $Z_0\subset Z$ be a closed subset of points where $M_{12}=0$.
A direct computation shows that the functions $(M_{12},A_i)$
(\ref{zz44}) -- (\ref{zz45}) are independent on an open
submanifold
\beq
U=Z\setminus Z_0 \label{zz47}
\eeq
of $Z$. At the same time, the functions $(\cH,M_{12},A_i)$ are
independent nowhere on $U$ because it follows from the expression
(\ref{zz50}) that
\beq
\cH=\frac{A^2-1}{2M^2} \label{zz52}
\eeq
on $U$ (\ref{zz47}). The well known dynamics of the Kepler system
shows that the Hamiltonian vector field of its Hamiltonian is
complete on $U$ (but not on Z).

The Poisson bracket of integrals of motion $M_{12}$ (\ref{zz44})
and $A_i$ (\ref{zz45}) obeys relations
\ben
&& \{M_{12},A_i\}=\eta_{2i}A_1 -\eta_{1i}A_2, \label{zz56}\\
&& \{A_1,A_2\}=2\cH M_{12}=\frac{A^2-1}{M_{12}}, \label{zz57}
\een
where $\eta_{ij}$ is an Euclidean metric on $\mathbb R^2$. It is
readily observed that these relations take the form (\ref{nc1}).
However, the matrix function $\mathbf{s}$ of the relations
(\ref{zz56}) -- (\ref{zz57}) fails to be of constant rank at
points where $\cH=0$. Therefore, let us consider the open
submanifolds $U_-\subset U$ where $\cH<0$ and $U_+$ where $\cH>0$.
Then we observe that the Kepler system with the Hamiltonian $\cH$
(\ref{zz41}) and the integrals of motion $(M_{ij},A_i)$
(\ref{zz44}) -- (\ref{zz45}) on $U_-$ and the Kepler system with
the Hamiltonian $\cH$ (\ref{zz41}) and the integrals of motion
$(M_{ij},A_i)$ (\ref{zz44}) -- (\ref{zz45}) on $U_+$ are
noncommutative CISs. Moreover, these CISs can be brought into the
form (\ref{zz60}) as follows.

Let us replace the integrals of motions $A_i$ with the integrals
of motion
\beq
L_i=\frac{A_i}{\sqrt{-2\cH}} \label{zz61}
\eeq
on $U_-$, and with the integrals of motion
\beq
K_i=\frac{A_i}{\sqrt{2\cH}} \label{zz62}
\eeq
on $U_+$.

The CIS $(M_{12},L_i)$ on $U_-$ obeys relations
\beq
\{M_{12},L_i\}=\eta_{2i}L_1 -\eta_{1i}L_2, \qquad
\{L_1,L_2\}=-M_{12}. \label{zz67}
\eeq
Let us denote $M_{i3}=-L_i$ and put the indexes
$\m,\nu,\al,\bt=1,2,3$. Then the relations (\ref{zz67}) are
brought into a form
\beq
\{M_{\m\nu},M_{\al\bt}\}=\eta_{\m\bt}M_{\nu\al} +
\eta_{\nu\al}M_{\m\bt} -
\eta_{\m\al}M_{\nu\bt}-\eta_{\nu\bt}M_{\m\al} \label{zz68}
\eeq
where $\eta_{\m\nu}$ is an Euclidean metric on $\mathbb R^3$. A
glance at the expression (\ref{zz68}) shows that the integrals of
motion $M_{12}$ (\ref{zz44}) and $L_i$ (\ref{zz61}) constitute a
Lie algebra $\cG=so(3)$. Its corank equals 1. Therefore the CIS
$(M_{12}, L_i)$ on $U_-$ is maximally integrable. The equality
(\ref{zz52}) takes a form
\beq
M^2 +L^2=-\frac1{2\cH}. \label{zz100}
\eeq

The CIS $(M_{12},K_i)$ on $U_+$ obeys relations
\beq
\{M_{12},K_i\}=\eta_{2i}K_1 -\eta_{1i}K_2, \qquad
\{K_1,K_2\}=M_{12}. \label{zz77}
\eeq
Let us denote $M_{i3}=-K_i$ and put the indexes
$\m,\nu,\al,\bt=1,2,3$. Then the relations (\ref{zz77}) are
brought into a form
\beq
\{M_{\m\nu},M_{\al\bt}\}=\rho_{\m\bt}M_{\nu\al} +
\rho_{\nu\al}M_{\m\bt} -
\rho_{\m\al}M_{\nu\bt}-\rho_{\nu\bt}M_{\m\al} \label{zz78}
\eeq
where $\rho_{\m\nu}$ is a pseudo-Euclidean metric of signature
$(+,+,-)$ on $\mathbb R^3$. A glance at the expression
(\ref{zz78}) shows that the integrals of motion $M_{12}$
(\ref{zz44}) and $K_i$ (\ref{zz62}) constitute a Lie algebra
$so(2,1)$. Its corank equals 1. Therefore the CIS $(M_{12}, K_i)$
on $U_+$ is maximally integrable. The equality (\ref{zz52}) takes
a form
\beq
K^2 -M^2=\frac1{2\cH}. \label{zz101}
\eeq

Thus, the Kepler problem on a phase space $\mathbb R^4$ falls into
two different maximally integrable systems on open submanifolds
$U_-$ and $U_+$ of $\mathbb R^4$. We agree to call them the Kepler
CISs on $U_-$ and $U_+$, respectively.

Let us study the first one, and let us put
\ben
&& F_1=-L_1, \qquad F_2=-L_2, \qquad F_3=-M_{12}, \label{zz102}\\
&& \{F_1,F_2\}=F_3, \qquad \{F_2,F_3\}=F_1, \qquad
\{F_3,F_1\}=F_2.\nonumber
\een
We have a fibred manifold
\beq
F: U_-\to N\subset\cG^*, \label{zz103}
\eeq
which is the momentum mapping to a Lie coalgebra $\cG^*=so(3)^*$,
endowed with the coordinates $(x_i)$ such that integrals of motion
$F_i$ on $\cG^*$ read $F_i=x_i$ (Remark \ref{nn240}). A base $N$
of the fibred manifold (\ref{zz103}) is an open submanifold of
$\cG^*$ given by a coordinate condition $x_3\neq 0$. It is a union
of two contractible components defined by conditions $x_3>0$ and
$x_3<0$. The coinduced Lie--Poisson structure on $N$ takes a form
\beq
w= x_2\dr^3\w\dr^1 + x_3\dr^1\w\dr^2 + x_1\dr^2\w\dr^3.
\label{j51b}
\eeq

The coadjoint action of $so(3)$ on $N$ reads
\beq
\ve_1=x_3\dr^2-x_2\dr^3, \quad \ve_2=x_1\dr^3-x_3\dr^1, \quad
\ve_3=x_2\dr^1-x_1\dr^2. \label{kk605}
\eeq
Orbits of this coadjoint action are given by an equation
\beq
x_1^2 + x_2^2 + x_3^2=\mathrm{const}. \label{kk606}
\eeq
They are level surfaces of a Casimir function
\be
C=x_1^2 + x_2^2 + x_3^2
\ee
and, consequently, the Casimir function
\beq
h=-\frac12(x_1^2 + x_2^2 + x_3^2)^{-1}. \label{zz120}
\eeq
A glance at the expression (\ref{zz100}) shows that the pull-back
$F^*h$ of this Casimir function (\ref{zz120}) onto $U_-$ is the
Hamiltonian $\cH$ (\ref{zz41}) of the Kepler system on $U_-$.

As was mentioned above, the Hamiltonian vector field of $F^*h$ is
complete. Furthermore, it is known that invariant submanifolds of
the Kepler CIS on $U_-$ are compact. Therefore, the fibred
manifold $F$ (\ref{zz103}) is a fibre bundle. Moreover, this fibre
bundle is trivial because $N$ is a disjoint union of two
contractible manifolds. Consequently, it follows from Theorem
\ref{cmp36} that the Kepler CIS on $U_-$ is global, i.e., it
admits global action-angle coordinates as follows.

The Poisson manifold $N$ (\ref{zz103}) can be endowed with the
coordinates
\beq
(I,x_1,\g), \qquad I<0, \qquad \g\neq\frac{\pi}2,\frac{3\pi}2,
\label{j52x}
\eeq
defined by the equalities
\ben
&& I=-\frac12(x_1^2 + x_2^2 + x_3^2)^{-1}, \label{j52}\\
&& x_2=\left(-\frac1{2I}-x_1^2\right)^{1/2}\sin\g, \qquad
x_3=\left(-\frac1{2I}-x_1^2\right)^{1/2}\cos\g. \nonumber
\een
It is readily observed that the coordinates (\ref{j52x}) are
Darboux coordinates of the Lie--Poisson structure (\ref{j51b}) on
$U_-$, namely,
\beq
w=\frac{\dr}{\dr x_1}\w \frac{\dr}{\dr \g}. \label{j53}
\eeq

Let $\vt_I$ be the Hamiltonian vector field of the Casimir
function $I$ (\ref{j52}). Its flows are invariant submanifolds of
the Kepler CIS on $U_-$ (Remark \ref{nn240}). Let $\al$ be a
parameter along the flow of this vector field, i.e.,
\beq
\vt_I= \frac{\dr}{\dr \al}. \label{zz121}
\eeq
Then $U_-$ is provided with the action-angle coordinates
$(I,x_1,\g,\al)$ such that the Poisson bivector associated to the
symplectic form $\Om_T$ on $U_-$ reads
\beq
W= \frac{\dr}{\dr I}\w \frac{\dr}{\dr \al} + \frac{\dr}{\dr x_1}\w
\frac{\dr}{\dr \g}. \label{j54}
\eeq
Accordingly, Hamiltonian vector fields of integrals of motion
$F_i$ (\ref{zz102}) take a form
\be
&& \vt_1= \frac{\dr}{\dr \g}, \\
&& \vt_2= \frac1{4I^2}\left(-\frac1{2I}-x_1^2\right)^{-1/2}\sin\g
\frac{\dr}{\dr \al} - x_1
\left(-\frac1{2I}-x_1^2\right)^{-1/2}\sin\g
\frac{\dr}{\dr \g} - \\
&& \qquad \left(-\frac1{2I}-x_1^2\right)^{1/2}\cos\g
\frac{\dr}{\dr x_1}, \\
&& \vt_3= \frac1{4I^2}\left(-\frac1{2I}-x_1^2\right)^{-1/2}\cos\g
\frac{\dr}{\dr \al} - x_1
\left(-\frac1{2I}-x_1^2\right)^{-1/2}\cos\g
\frac{\dr}{\dr \g} + \\
&&\qquad \left(-\frac1{2I}-x_1^2\right)^{1/2}\sin\g \frac{\dr}{\dr x_1}.
\ee
A glance at these expressions shows that the vector fields $\vt_1$
and $\vt_2$ fail to be complete on $U_-$ (Remark \ref{zz90}).

One can say something more about the angle coordinate $\al$. The
vector field $\vt_I$ (\ref{zz121}) reads
\be
\frac{\dr}{\dr\al}= \op\sum_i\left(\frac{\dr \cH}{\dr
p_i}\frac{\dr}{\dr q_i}-\frac{\dr \cH}{\dr q_i}\frac{\dr}{\dr
p_i}\right).
\ee
This equality leads to relations
\be
\frac{\dr q_i}{\dr \al}=\frac{\dr \cH}{\dr p_i}, \qquad \frac{\dr
p_i}{\dr \al}=-\frac{\dr \cH}{\dr q_i},
\ee
which take a form of the Hamilton equation. Therefore, the
coordinate $\al$ is a cyclic time $\al=t\,\mathrm{mod}\,2\pi$
given by the well-known expression
\be
&&\al=\f-a^{3/2}e\sin(a^{-3/2}\f),\qquad
r=a(1-e\cos(a^{-3/2}\f)),\\
&& a=-\frac1{2I}, \qquad e=(1+2IM^2)^{1/2}.
\ee

Now let us turn to the Kepler CIS on $U_+$. It is a globally
integrable system with noncompact invariant submanifolds as
follows.

Let us put
\ben
&& S_1=-K_1, \qquad S_2=-K_2, \qquad S_3=-M_{12}, \label{zz102a}\\
&& \{S_1,S_2\}=-S_3, \qquad \{S_2,S_3\}=S_1, \qquad \{S_3,S_1\}=S_2.
\nonumber
\een
We have a fibred manifold
\beq
S: U_+\to N\subset\cG^*, \label{zz103a}
\eeq
which is the momentum mapping to a Lie coalgebra
$\cG^*=so(2,1)^*$, endowed with the coordinates $(x_i)$ such that
integrals of motion $S_i$ on $\cG^*$ read $S_i=x_i$. A base $N$ of
the fibred manifold (\ref{zz103a}) is an open submanifold of
$\cG^*$ given by a coordinate condition $x_3\neq 0$. It is a union
of two contractible components defined by conditions $x_3>0$ and
$x_3<0$. The coinduced Lie--Poisson structure on $N$ takes a form
\beq
w= x_2\dr^3\w\dr^1 - x_3\dr^1\w\dr^2 + x_1\dr^2\w\dr^3.
\label{j51a}
\eeq

The coadjoint action of $so(2,1)$ on $N$ reads
\be
\ve_1=-x_3\dr^2-x_2\dr^3, \qquad \ve_2=x_1\dr^3+x_3\dr^1, \qquad
\ve_3=x_2\dr^1-x_1\dr^2.
\ee
The orbits of this coadjoint action are given by an equation
\be
x_1^2 + x_2^2 - x_3^2=\mathrm{const}.
\ee
They are the level surfaces of the Casimir function
\be
C=x_1^2 + x_2^2 - x_3^2
\ee
and, consequently, the Casimir function
\beq
h=\frac12(x_1^2 + x_2^2 - x_3^2)^{-1}. \label{zz120a}
\eeq
A glance at the expression (\ref{zz101}) shows that the pull-back
$S^*h$ of this Casimir function (\ref{zz120a}) onto $U_+$ is the
Hamiltonian $\cH$ (\ref{zz41}) of the Kepler system on $U_+$.

As was mentioned above, the Hamiltonian vector field of $S^*h$ is
complete. Furthermore, it is known that invariant submanifolds of
the Kepler CIS on $U_+$ are diffeomorphic to $\mathbb R$.
Therefore, the fibred manifold $S$ (\ref{zz103a}) is a fibre
bundle. Moreover, this fibre bundle is trivial because $N$ is a
disjoint union of two contractible manifolds. Consequently, it
follows from Theorem \ref{cmp36} that the Kepler CIS on $U_+$ is
globally integrable, i.e., it admits global action-angle
coordinates as follows.

The Poisson manifold $N$ (\ref{zz103a}) can be endowed with the
coordinates
\be
(I,x_1,\la), \qquad I>0, \qquad \la\neq 0,
\ee
defined by the equalities
\be
&& I=\frac12(x_1^2 + x_2^2 - x_3^2)^{-1}, \\
&& x_2=\left(\frac1{2I}-x_1^2\right)^{1/2}\cosh\la, \qquad
x_3=\left(\frac1{2I}-x_1^2\right)^{1/2}\sinh\la.
\ee
These coordinates are Darboux coordinates of the Lie--Poisson
structure (\ref{j51a}) on $N$, namely,
\beq
w=\frac{\dr}{\dr \la}\w \frac{\dr}{\dr x_1}. \label{j53a}
\eeq

Let $\vt_I$ be the Hamiltonian vector field of the Casimir
function $I$ (\ref{j52}). Its flows are invariant submanifolds of
the Kepler CIS on $U_+$ (Remark \ref{nn240}). Let $\tau$ be a
parameter along the flows of this vector field, i.e.,
\beq
\vt_I= \frac{\dr}{\dr \tau}. \label{zz121a}
\eeq
Then $U_+$ (\ref{zz103a}) is provided with the action-angle
coordinates $(I,x_1,\la,\tau)$ such that the Poisson bivector
associated to the symplectic form $\Om_T$ on $U_+$ reads
\beq
W= \frac{\dr}{\dr I}\w \frac{\dr}{\dr \tau} + \frac{\dr}{\dr
\la}\w \frac{\dr}{\dr x_1}. \label{j54a}
\eeq
Accordingly, Hamiltonian vector fields of integrals of motion
$S_i$ (\ref{zz102a}) take a form
\be
&& \vt_1= -\frac{\dr}{\dr \la}, \\
&& \vt_2= \frac1{4I^2}\left(\frac1{2I}-x_1^2\right)^{-1/2}\cosh\la
\frac{\dr}{\dr \tau} + x_1
\left(\frac1{2I}-x_1^2\right)^{-1/2}\cosh\la
\frac{\dr}{\dr \la} + \\
&& \qquad \left(\frac1{2I}-x_1^2\right)^{1/2}\sinh\la
\frac{\dr}{\dr x_1}, \\
&& \vt_3= \frac1{4I^2}\left(\frac1{2I}-x_1^2\right)^{-1/2}\sinh\la
\frac{\dr}{\dr \tau} + x_1
\left(\frac1{2I}-x_1^2\right)^{-1/2}\sinh\la
\frac{\dr}{\dr \la} + \\
&&\qquad \left(\frac1{2I}-x_1^2\right)^{1/2}\cosh\la \frac{\dr}{\dr x_1}.
\ee

Similarly to the angle coordinate $\al$ (\ref{zz121}), the angle
coordinate $\tau$ (\ref{zz121a}) obeys the Hamilton equation
\be
\frac{\dr q_i}{\dr \tau}=\frac{\dr \cH}{\dr p_i}, \qquad \frac{\dr
p_i}{\dr \tau}=-\frac{\dr \cH}{\dr q_i}.
\ee
Therefore, it is the time $\tau=t$ given by the well-known
expression
\be
&& \tau=s-a^{3/2}e\sinh (a^{-3/2}s),\qquad r=a(e\cosh
(a^{-3/2}s)-1),\\
&& a=\frac1{2I}, \qquad e=(1+2IM^2)^{1/2}.
\ee

\end{document}